\documentclass[11pt]{article}
\usepackage{amsfonts}
\usepackage{amsmath}
\usepackage{amssymb}
\usepackage{amsthm}
\usepackage{enumitem}
\usepackage{multirow}
\usepackage[margin=1in]{geometry}
\usepackage{graphicx}
\usepackage{hyperref}
\hypersetup{
    colorlinks=true,
    linkcolor=blue,
    urlcolor=blue,
    citecolor = blue
}
\usepackage[utf8]{inputenc}
\usepackage{bbm}
\usepackage[numbers]{natbib}
\usepackage[dvipsnames]{xcolor}

\usepackage[boxruled,vlined,nofillcomment]{algorithm2e}
	\SetKwProg{Fn}{}{\string:}{}
	\SetKwFor{While}{While}{}{}
	\SetKwFor{For}{For}{}{}
	\SetKwIF{If}{ElseIf}{Else}{If}{:}{ElseIf}{Else}{:}
	\SetKw{Return}{Return}
	\SetKw{Parameters}{Parameters:}
	\LinesNumbered

\title{Connecting Robust Shuffle Privacy and Pan-Privacy}
\author{Victor Balcer\thanks{Harvard University, \href{mailto:vbalcer@g.harvard.edu}{\texttt{vbalcer@g.harvard.edu}}. Supported by NSF grant CNS-1565387}
\and Albert Cheu\thanks{Northeastern University, \href{mailto:cheu.a@husky.neu.edu}{\texttt{cheu.a@husky.neu.edu}}. Supported by NSF grants CCF-1718088, CCF-
1750640, and CNS-1816028.}
\and Matthew Joseph\thanks{Google New York,
\href{mailto:mtjoseph@google.com}{\texttt{mtjoseph@google.com}}. Part of this work done while a graduate student at the University of Pennsylvania}
\and Jieming Mao\thanks{Google New York, \href{mailto:maojm@google.com}{\texttt{maojm@google.com}}.}}

\newtheorem{theorem}{Theorem}[section]
\newtheorem{fact}[theorem]{Fact}
\newtheorem{lemma}[theorem]{Lemma}
\newtheorem{corollary}[theorem]{Corollary}
\newtheorem{claim}[theorem]{Claim}

\theoremstyle{definition}
\newtheorem{definition}[theorem]{Definition}

\newcommand{\bD}{\mathbf{D}}

\newcommand{\bT}{\mathbf{T}}
\newcommand{\bU}{\mathbf{U}}

\newcommand{\cA}{\mathcal{A}}

\newcommand{\cH}{\mathcal{H}}

\newcommand{\cM}{\mathcal{M}}

\newcommand{\cP}{\mathcal{P}}
\newcommand{\cQ}{\mathcal{Q}}
\newcommand{\cR}{\mathcal{R}}
\newcommand{\cS}{\mathcal{S}}

\newcommand{\cX}{\mathcal{X}}
\newcommand{\cY}{\mathcal{Y}}
\newcommand{\cZ}{\mathcal{Z}}

\newcommand{\Ber}[1]{\mathbf{Ber}\left(#1\right)}
\newcommand{\tBer}[1]{\mathbf{Ber}(#1)}
\newcommand{\Bin}{\mathbf{Bin}}
\newcommand{\Lap}[1]{\mathbf{Lap}(#1)}
\newcommand{\Pois}{\mathbf{Pois}}
\newcommand{\SG}{\mathbf{SG}}

\newcommand{\eps}{\varepsilon}
\newcommand{\ones}{\!\vec{\,1}}
\newcommand{\zo}{\{0,1\}}

\newcommand{\comments}{0	}
\newcommand{\vb}[1]{\ifnum\comments=1\textcolor{violet}{[VB: #1]}\else \fi}
\newcommand{\ac}[1]{\ifnum\comments=1\textcolor{orange}{[AC: #1]}\else \fi}
\newcommand{\mj}[1]{\ifnum\comments=1\textcolor{teal}{[MJ: #1]}\else \fi}
\newcommand{\jm}[1]{\ifnum\comments=1\textcolor{brown}{[JM: #1]}\else \fi}

\newcommand{\de}{\mathsf{DE}}

\newcommand{\poly}{\mathsf{poly}}

\newcommand{\A}{\mathcal{A}}

\newcommand{\In}{\mathcal{I}}

\newcommand{\Ou}{\mathcal{O}}

\newcommand{\Z}{\mathbb{Z}}

\newcommand{\E}[2]{\mathbb{E}_{#1}\left[ #2 \right]}
\newcommand{\ex}[2]{{\ifx&#1& \mathbb{E} \else \underset{#1}{\mathbb{E}} \fi \left[#2\right]}}

\newcommand{\half}{\ensuremath{\frac{1}{2}}}

\newcommand{\mde}{\mathsf{DE}}
\newcommand{\mor}{\mathsf{OR}}
\newcommand{\mmod}{\mathsf{MOD}}
\newcommand{\N}{\mathbb{N}}

\renewcommand{\P}[2]{\mathbb{P}_{#1}\left[#2\right]}
\newcommand{\R}{\mathbb{R}}

\newcommand{\tc}{\tilde c}
\newcommand{\tv}[2]{\|#1 - #2\|_\mathrm{TV}}
\newcommand{\ut}{\mathsf{UT}}
\newcommand{\Var}[2]{\text{Var}_{#1}\left[#2\right]}
\newcommand{\zsum}{\mathsf{ZSUM}}

\begin{document}

\maketitle

\begin{abstract}
In the \emph{shuffle model} of differential privacy, data-holding users send randomized messages to a secure shuffler, the shuffler permutes the messages, and the resulting collection of messages must be differentially private with regard to user data. In the \emph{pan-private} model, an algorithm processes a stream of data while maintaining an internal state that is differentially private with regard to the stream data. We give evidence connecting these two apparently different models.

Our results focus on \emph{robustly} shuffle private protocols, whose privacy guarantees are not greatly affected by malicious users. First, we give robustly shuffle private protocols and upper bounds for counting distinct elements and uniformity testing. Second, we use pan-private lower bounds to prove robustly shuffle private lower bounds for both problems. Focusing on the dependence on the domain size $k$, we find that robust approximate shuffle privacy and approximate pan-privacy have additive error $\Theta(\sqrt{k})$ for counting distinct elements. For uniformity testing, we give a robust approximate shuffle private protocol with sample complexity $\tilde O(k^{2/3})$ and show that an $\Omega(k^{2/3})$ dependence is necessary for any robust pure shuffle private tester. Finally, we show that this connection is useful in both directions: we give a pan-private adaptation of recent work on shuffle private histograms and use it to recover further separations between pan-privacy and interactive local privacy.
\end{abstract}

\section{Introduction}
\label{sec:intro}
Differential privacy~\cite{DMNS06} guarantees that an algorithm's output is quantifiably insensitive to small changes in its input. This insensitivity ensures that differentially private algorithms are not greatly affected by any one data point, which in turn provides privacy for data contributors. The basic differential privacy framework is the foundation for many models, including central~\cite{DMNS06}, local~\cite{DMNS06, BNO08, KLNRS11}, pan-~\cite{DNPRY10}, blended~\cite{AKZHL17}, and most recently shuffle~\cite{CSUZZ19} privacy. 

This work focuses on pan-privacy and shuffle privacy. A pan-private algorithm receives a stream of raw data and processes it one element at a time. After seeing each element, the algorithm updates its internal state to incorporate information from the new element and then proceeds to the next element in the stream. At the end of the stream, the algorithm processes its final internal state to extract and output useful information. Privacy constrains the internal state and output to be differentially private functions of the stream: changing one element of the stream must not greatly affect the joint distribution of any one internal state and the final output.

In shuffle privacy, there is no algorithm receiving a stream of data. Instead, the data is distributed, and  each user holds a single data point. The users follow a prescribed protocol where each employs a randomizer function to compute messages based on their data point and then sends these messages to a secure shuffler\footnote{ In practice, this may be a cryptographic protocol for multi-party shuffling rather than a trusted third-party shuffler.}. The shuffler randomly permutes the messages and releases the shuffled collection of messages publicly. Finally, an analyzer processes the public shuffler output to extract useful information. Here, privacy constrains the shuffler's output: changing one user's data point must not greatly affect the distribution for the messages released by the shuffler. Privacy is over the random coins of the users' randomizers and the shuffler. Since the analyzer only post-processes the shuffler's output, its actions do not affect users' privacy guarantees (see Fact~\ref{fact:post} for details on post-processing).

One possible complication of the shuffle model is that, without further restrictions, the privacy guarantees of shuffle private protocols are not necessarily robust to malicious users. For example, it is possible for a shuffle private protocol to place all responsibility for ``noisy messages'' on a single user. In that case, compromising that single user would destroy the privacy guarantee for \emph{all} users of the protocol. To avoid this weakness, we focus on protocols that satisfy \emph{robust} shuffle privacy. These protocols still guarantee privacy for honest users in the presence of (a limited number of) malicious users.

\subsection{Our Contributions}
We give several results connecting robust shuffle privacy and pan-privacy. Details and comparisons appear in Figure~\ref{fig:results}. 

\begin{enumerate}
	\item We construct a protocol for counting distinct elements that satisfies robust approximate shuffle privacy and obtains additive error $O(\sqrt{k}/\eps)$ (Theorem~\ref{thm:de_upper}). We then strengthen and adapt a lower bound from pan-privacy to show that $\Omega(\sqrt{k/\eps})$ additive error is necessary for robust approximate shuffle privacy when the number of users $n \geq k$ (Theorem~\ref{thm:de-lower-bound}). In contrast, adding $\Lap{1/\eps}$ noise to the true distinct count guarantees $O(1/\eps)$ error in the central model.
	\item We construct a protocol for uniformity testing that satisfies robust approximate shuffle privacy\footnote{An early version of this paper incorrectly claimed a tester that satisfies robust \emph{pure} shuffle privacy.} with sample complexity dependence on $k$ of $\tilde O(k^{2/3})$ (Theorem~\ref{thm:ut-app-upper-bound}). We again adapt a lower bound from pan-privacy to show that any protocol satisfying robust \emph{pure} shuffle privacy requires $\Omega(k^{2/3})$ samples (Theorem~\ref{thm:ut-pure-lower-bound}). These two bounds are not directly comparable --- the lower bound requires conditions that the upper bound does not meet --- but they offer partial evidence for the general connection between the two models.
	\item We show how to adapt recent work on shuffle private histograms~\cite{BC20} for pan-privacy (Theorem~\ref{theorem:hist-pan}). As a corollary, pan-privacy inherits the same separations from interactive local privacy as the shuffle model for ``support identification'' problems.
\end{enumerate}

\begin{figure*}[h!]
\renewcommand*{\arraystretch}{1.8}
\makebox[\textwidth][c]{
\begin{tabular}{|c|c|c|c|}
 \hline
 \multirow{2}{*}{Privacy Type} & Histograms & $\alpha$-Uniformity Testing & Distinct Elements \\
& ($\ell_\infty$ error) & (sample complexity) & (additive error) \\ \hline \hline
Central  & $\Theta\big(\frac{1}{\eps}\log\big(\min\big(\frac{1}{\delta}, k\big)\big)\big)$ ~\cite{BNS16, HT10, BBKN10} & $\Theta\Big(\frac{\sqrt{k}}{\alpha^2} + \frac{\sqrt{k}}{\alpha\sqrt{\eps}} + \frac{k^{1/3}}{\alpha^{4/3}\eps^{2/3}} + \frac{1}{\alpha\eps}\Big)$ \cite{ASZ18} & $\Theta\big(\frac{1}{\eps}\big)$~\cite{DMNS06} * \vphantom{\bigg(}\\
\hline
\multirow{2}{*}{Pan} & \multirow{2}{*}{$O\big(\frac{1}{\eps^2}\log\big(\frac{1}{\delta}\big)\big)$}  & $O\Big(\frac{k^{2/3}}{\alpha^{4/3} \eps^{2/3}} + \frac{\sqrt{k}}{\alpha^2} + \frac{\sqrt{k}}{\alpha \eps}\Big)$ \cite{AJM19} * & $O\!\left(\frac{\sqrt{k}}{\eps}\right)$ \cite{DNPRY10} * \\
 & & $\Omega\Big(\frac{k^{2/3}}{\alpha^{4/3} \eps^{2/3}} + \frac{\sqrt{k}}{\alpha^2}  + \frac{1}{\alpha\eps}\Big)$ \cite{AJM19} * & $\Omega\Big(\tfrac{\sqrt{k}}{\sqrt{\eps}}\Big)$ \vphantom{\bigg(} \\
\hline
\multirow{2}{*}{Robust Shuffle} & \multirow{2}{*}{$O\big(\frac{1}{\eps^2}\log\big(\frac{1}{\delta}\big)\big)$ \cite{BC20}}  & $O\Big(\Big[\frac{k^{2/3}}{\alpha^{4/3}\eps^{2/3}} + \frac{\sqrt{k}}{\alpha^2} + \frac{\sqrt{k}}{\alpha \eps}\Big]\sqrt{\log\big(\tfrac{k}{\delta}\big)}\Big)$ & $O\Big(\frac{\sqrt{k}}{\eps}\Big)$ \\
 & & $\Omega\Big(\frac{k^{2/3}}{\alpha^{4/3} \eps^{2/3}} + \frac{\sqrt{k}}{\alpha^2} +  \frac{1}{\alpha\eps}\Big)$ * & $\Omega\Big(\tfrac{\sqrt{k}}{\sqrt{\eps}}\Big)$ \vphantom{\bigg(} \\
\hline
\end{tabular}}
\caption{Overview of our main results given a data domain of size $k$, $\eps=O(1)$ and $\delta<1/n<\eps$. All bounds hold with constant probability. Results marked by * hold when $\delta=0$. Uncited results are new in this work.}
\label{fig:results}
\end{figure*}

\subsection{Related Work}
Dwork, Naor, Pitassi, Rothblum, and Yekhanin~\cite{DNPRY10} introduced pan-privacy and, for a data domain of size $k$, showed how to pan-privately count distinct elements to additive accuracy $O(\sqrt{k}/\eps)$. Mir, Muthukrishnan, Nikolov, and Wright~\cite{MMNW11} constructed a low-memory analogue with the same accuracy guarantee and a matching (for $n \geq k$) lower bound based on linear program decoding. Both works employ \emph{user-level} pan-privacy, where one user may contribute many elements to the stream, and a neighboring stream may replace all of a user's contributions. We instead imitate Amin, Joseph, and Mao~\cite{AJM19} and study \emph{record-level} pan-privacy. Here, neighboring streams differ in at most one stream element. For uniformity testing,~\citet{AJM19} gave tight (in the domain size $k$) $\Theta(k^{2/3})$ sample complexity bounds for pure pan-privacy. Acharya, Sun, and Zhang~\cite{ASZ18} showed that $\Theta(\sqrt{k})$ is the optimal dependence under approximate central privacy, and $\Theta(k)$ is optimal for the strictly more private model of sequentially interactive local privacy~\cite{ACFT19, AJM19} (see Section~\ref{sec:prelims} for a definition of local privacy).

Building on the empirical work of Bittau, Erlingsson, Maniatis, Mironov, Raghunathan, Lie, Rudominer, Kode, Tinnes, and Seefeld~\cite{BEMMR+17}, Cheu, Smith, Ullman, Zeber, and Zhilyaev~\cite{CSUZZ19} and Erlingsson, Feldman, Mironov, Raghunathan, Talwar, and Thakurta~\cite{EFMRTT19} independently and simultaneously introduced different formal definitions of shuffle privacy.~\citet{CSUZZ19} defined a model where a single noninteractive batch of messages is shuffled, while~\citet{EFMRTT19} considered shuffled \emph{users} who may participate in an interactive protocol. We follow most of the shuffle privacy literature and use the former variant. We also focus on the multi-message model where each user may send multiple messages to the shuffler. For the problem of shuffle private bit summation, a line of papers~\cite{CSUZZ19, BBGN19, BBGN20, GGKMPV20} has obtained $O(1/\eps)$ accuracy tight with that of the central model. For the problem of computing histograms, the best known shuffle private $\ell_\infty$ guarantee differs by a $1/\eps$ factor from the centrally private guarantee~\cite{GGKPV19, BC20}. 

We briefly discuss past work touching on robust shuffle privacy.~\citet{CSUZZ19} showed that any single-message $(\eps,\delta)$-shuffle private protocol on $n$ users also grants $(\eps + \ln(n), \delta)$-\emph{local privacy}. This is a form of robustness, as it guarantees privacy even when all other users are malicious. However, it only holds for single-message protocols, which are strictly weaker than multi-message protocols. Along similar lines, Balcer and Cheu~\cite{BC20} showed that any single-message $\eps$-shuffle private protocol is also $\eps$-locally private. Balle, Bell, Gasc\'on and Nissim~\cite{BBGN20} discussed the effect of malicious users on the accuracy guarantee of a shuffle protocol. In contrast, our notion of robust shuffle privacy guarantees privacy even with (a controlled fraction of) malicious users in multi-message protocols.

\subsection{Organization}
Basic definitions appear in the Preliminaries (Section~\ref{sec:prelims}). Further specific preliminaries appear in their corresponding sections. We start with distinct elements (Section~\ref{sec:distinct}), then cover uniformity testing (Section~\ref{sec:uniform}) and histograms (Section~\ref{sec:histogram}). We conclude with a brief discussion and some further questions in Section~\ref{sec:conc}.

\section{Preliminaries}
\label{sec:prelims}
Throughout this work, we use $[k] := \{1, 2, \ldots, k\}$ and $\mathbb{N} := \{1, 2, \ldots \}$.
%, $\mathbb{Z}_k = \{0, 1, \ldots, k\}$, and $\mathbb{Z}_{\geq 0} = \mathbb{N} \cup \{0\}$.

\subsection{Differential Privacy}
\label{sec:prelims_dp}
We define a dataset $\vec{x} \in \cX^n$ to be an ordered tuple of $n$ rows where each row is drawn from a data universe $\cX$ and corresponds to the data of one user. Two datasets $\vec{x},\vec{x}\,' \in \cX^n$ are considered \emph{neighbors} (denoted as $\vec{x} \sim \vec{x}\,'$) if they differ in at most one row.

\begin{definition}[Differential Privacy \cite{DMNS06}]
An algorithm $\cM: \cX^n \rightarrow \cZ$ satisfies \emph{$(\eps, \delta)$-differential privacy} if, for every pair of neighboring datasets $\vec{x}$ and $\vec{x}\,'$ and every subset $T \subset \cZ$,
	$$\P{}{\cM(\vec{x}\vphantom{'}) \in T} \le e^\eps \cdot \P{}{\cM(\vec{x}\,') \in T} + \delta.$$
When $\delta > 0$, we say $\cM$ satisfies \emph{approximate differential privacy}. When $\delta = 0$, we say $\cM$ satisfies \emph{pure differential privacy} and we omit the $\delta$ parameter.
\end{definition}

Because this definition assumes that the algorithm $\cM$ has ``central'' access to compute on the entire raw dataset, we sometimes call this \emph{central} differential privacy. Two common facts about differentially privacy will be useful; proofs appear in Chapter 2 of the survey by Dwork and Roth~\cite{DR14}. First, privacy is preserved under post-processing. 

\begin{fact}
\label{fact:post}
	For $(\eps,\delta)$-differentially private algorithm $\A : \cX^n \to \cZ$ and arbitrary random function $f : \cZ \to \cZ'$, $f \circ \A$ is $(\eps,\delta)$-differentially private.
\end{fact}

This means that any computation based solely on the output of a differentially private function does not affect the privacy guarantee. Second, privacy composes neatly.

\begin{fact}
\label{fact:comp}
	For $(\eps_1, \delta_1)$-differentially private $\A_1$ and $(\eps_2, \delta_2)$-differentially private $\A_2$, $\A_3$ defined by $\A_3(D) = (\A_1(D), \A_2(D))$ is $(\eps_1 + \eps_2, \delta_1 + \delta_2)$-differentially private.
\end{fact}

One useful centrally private algorithm is the \emph{binomial mechanism}. We state the version given by Ghazi, Golowich, Kumar, Pagh, and Velingker~\cite{GGKPV19}.
%\vb{why not just define general version since we it's stated for the values using GGKPV?}
% \begin{lemma}[\cite{DKMMN06}]
% \label{lem:b_noise}
	% Let $f \colon \cX^n \to \Z$ be a 1-sensitive function, i.e. $|f(\vec{x}) - f(\vec{x}\,')| \leq 1$ for all neighboring datasets $\vec{x}, \vec{x}\,' \in \cX^n$. Then for $\eps > 0$, $\delta \in (0,1)$, and $\lambda \geq 20\left(\tfrac{e^\eps+1}{e^\eps-1}\right)^2\ln\left(\tfrac{2}{\delta}\right)$, the algorithm that samples $\eta \sim \Bin(\lambda, 1/2)$ then outputs $f(\vec{x}) + \eta$ is $(\eps,\delta)$-differentially private.
% \end{lemma}
\begin{lemma}[Binomial Mechanism \cite{DKMMN06,GGKPV19}]
\label{lem:b_noise}
	Let $f \colon \cX^n \to \Z$ be a 1-sensitive function, i.e. $|f(\vec{x}) - f(\vec{x}\,')| \leq 1$ for all neighboring datasets $\vec{x}, \vec{x}\,' \in \cX^n$. Fix any $\ell\in\N$, $p\in(0,1)$,  $\eps > 0$, and $\delta \in (0,1)$ such that $$\ell \cdot \min(p,1-p) \geq 10\cdot \left( \tfrac{e^\eps+1}{e^\eps-1}\right)^2 \cdot  \ln\left(\tfrac{2}{\delta} \right).$$ The algorithm that samples $\eta \sim \Bin(\ell, p)$ and outputs $f(\vec{x}) + \eta$ is $(\eps,\delta)$-differentially private. The error is $O(\frac{1}{\eps}\sqrt{\log \frac{1}{\delta}})$.
\end{lemma}

\iffalse
We will also use the \emph{geometric mechanism}. Let $\SG(\eps)$ denote the symmetric geometric distribution with parameter $\eps$. For every $v\in \Z$, it has mass $e^{-\eps|v|} \cdot \tfrac{e^\eps-1}{e^\eps+1}$, and may be viewed as a discrete analogue of the Laplace distribution with mean 0.

\begin{lemma}[\cite{GRS12}]
\label{lem:sg_noise}
	Let $f \colon \cX^n \to \Z$ be a 1-sensitive function, i.e.  $|f(\vec{x}) - f(\vec{x}\,')| \leq 1$ for all neighboring datasets $\vec{x}, \vec{x}\,' \in \cX^n$. Then for $\eps > 0$, the algorithm that samples $\eta \sim \SG(\eps)$ then outputs $f(\vec{x}) + \eta$ is $\eps$-differentially private.
\end{lemma}
\fi

\subsection{Pan-privacy}
\label{sec:prelims_pan}
\emph{Pan-privacy} is defined for a different setting. Unlike centrally private algorithms, pan-private algorithms are online: they receive raw data one element at a time in a stream. At each step in the stream, the algorithm receives a data point, update its internal state based on this data point, and then proceeds to the next element. The only way the algorithm ``remembers'' past elements is through its internal state. As in the case of datasets, we say that two streams $\vec{x}$ and $\vec{x}\,'$ are \emph{neighbors} if they differ in at most one element. Pan-privacy requires the algorithm's internal state and output to be differentially private with regard to neighboring streams.

\begin{definition}[Online Algorithm]
An \emph{online algorithm} $\cQ$ is defined by an internal algorithm $\cQ_{\In}$ and an output algorithm $\cQ_{\Ou}$. $\cQ$ processes a stream of elements through repeated application of $\cQ_{\In} \colon \cX \times \In \to \In$, which (with randomness) maps a stream element and internal state to an internal state. At the end of the stream, $\cQ$ publishes a final output by executing $\cQ_{\Ou}$ on its final internal state. %based on its final internal state, $\cQ_{\Ou(i)}$. 
\end{definition}

\begin{definition}[Pan-privacy \cite{DNPRY10, AJM19}]
\label{def:pan}
	Given an online algorithm $\cQ$, let $\cQ_{\In}(\vec{x})$ denote its internal state after processing stream $\vec{x}$, and let $\vec{x}_{\leq t}$ be the first $t$ elements of $\vec{x}$. We say $\cQ$ is \emph{$(\eps,\delta)$-pan-private} if, for every pair of neighboring streams $\vec{x}$ and $\vec{x}\,'$, every time $t$ and every set of internal state, output state pairs $T \subset \In \times \Ou$,
    \begin{equation}
    \label{eq:pan}
    \P{\cQ}{\big(\cQ_{\In}(\vec{x}_{\leq t}), \cQ_{\Ou}(\cQ_{\In}(\vec{x}))\big) \in T} \leq e^\eps\cdot\P{\cQ}{\big(\cQ_{\In}(\vec{x}\,'_{\!\leq t}), \cQ_{\Ou}(\cQ_{\In}(\vec{x}\,'))\big) \in T} + \delta.
    \end{equation}
    When $\delta = 0$, we say $\cQ$ is \emph{$\eps$-pan-private}.
\end{definition}

Taken together, these requirements protect against an adversary that sees any one internal state of $\cQ$ as well as its final output\footnote{As shown by~\citet{AJM19}, the precise assumption of one internal state intrusion is necessary. Pan-privacy against multiple internal state intrusions collapses to the much stronger notion of local privacy.}. By the output requirement, any pan-private algorithm also satisfies central differential privacy. The key additional contribution of pan-privacy is the maintenance of the differentially private internal state. This strengthens the central privacy guarantee by protecting data contributors against future events. For example, a user may trust the current algorithm operator but wish to protect themselves against the possibility that the operator will be acquired or subpoenaed in the future. Under pan-privacy, post-processing (Fact~\ref{fact:post}) ensures that future views of the pan-private algorithm's state will be differentially private with respect to past data. 

Our definition of pan-privacy is the specific variant given by~\citet{AJM19}. This version guarantees record-level (uncertainty about the presence of any single stream element) rather than user-level (uncertainty about the presence of any one data universe element) privacy. We use this variant because, like the shuffle model, we assume each data contributor has a single data point.

\subsection{Shuffle Privacy}
\label{sec:prelims_shuffle}
The shuffle model of differential privacy views the dataset as a distributed object where each of $n$ users holds a single row. Each user provides their data point as input to a randomizing function and securely submits the resulting randomized outputs to a shuffler. The shuffler permutes the users' outputs and releases the shuffled messages. It is this collection of messages that needs to satisfy differential privacy: altering one user's data point must not greatly change the distribution of the shuffled messages. 

In this way, the shuffle model strengthens the privacy guarantee that users receive. Here, users need only trust that (1) there is a secure way to shuffle the randomized messages\footnote{A more detailed description of a shuffler and its implementation details appears in the work of~\citet{BEMMR+17}.} and (2) sufficiently many users follow the protocol. In contrast, central differential privacy requires users to trust a third party algorithm operator to securely store and compute on the raw data. 

\begin{definition}[Shuffle Model \cite{BEMMR+17, CSUZZ19}]
A protocol $\cP$ in the \emph{shuffle model} consists of three randomized algorithms:
\begin{itemize}
\item
    A \emph{randomizer} $\cR: \cX \rightarrow \cY^*$ mapping data to (possibly variable-length) vectors. The length of the vector is the number of messages sent. If, on all inputs, the probability of sending a single message is 1, then the protocol is said to be \emph{single-message}. Otherwise, the protocol is \emph{multi-message}.
\item
    A \emph{shuffler} $\cS: \cY^* \rightarrow \cY^*$ that concatenates message vectors and then applies a uniformly random permutation to the messages. 
\item
    An \emph{analyzer} $\cA: \cY^* \rightarrow \cZ$ that computes on a permutation of messages.
\end{itemize}
As $\cS$ is the same in every protocol, we identify each shuffle protocol by $\cP = (\cR, \cA)$. We define its execution on input $\vec{x}\in\cX^n$ as
$$
\cP(\vec{x}) := \cA(\cS(R(x_1), \ldots, R(x_n))).
$$
We assume that $\cR$ and $\cA$ have access to $n$ and an arbitrary amount of public randomness.
\end{definition}

With this setup, we use the following definition of shuffle differential privacy.
\begin{definition} [Shuffle Differential Privacy \cite{CSUZZ19}]
\label{def:shuffle_dp}
	A protocol $\cP = (\cR, \cA)$ is \emph{$(\eps, \delta)$-shuffle differentially private} if, for all $n\in \N$, the algorithm $(\cS \circ \cR^n) := \cS(\cR(x_1), \ldots, \cR(x_n))$ is $(\eps, \delta)$-differentially private. The privacy guarantee is over the internal randomness of the users' randomizers and not the public randomness of the shuffle protocol.
\end{definition}

For brevity, we typically call these protocols ``shuffle private.'' We now sketch an example of a shuffle private protocol. Consider the setting where each user $i$'s data point is two bits $(x_{i,1},x_{i,2})$ and the goal is to sum both across users, $\sum_{i=1}^n x_{i,1}$ and $\sum_{i=1}^n x_{i,2}$. A simple randomizer independently flips each bit with probability $p$ and outputs the resulting two randomized bits. When there are $n$ users, the shuffler takes in $2n$ messages and outputs one of the $(2n)!$ permutations of the messages uniformly at random. The analyzer can then recover unbiased estimates by rescaling according to the randomization probability $p$ and the number of users $n$. Lemma \ref{lem:b_noise} implies that setting $p \approx \tfrac{\log(1/\delta)}{\eps^2n}$ guarantees $(\eps,\delta)$-shuffle privacy, so the analyzer may recover estimates with error $\approx \tfrac{\sqrt{\log(1/\delta)}}{\eps}$.

Note, however, that Definition~\ref{def:shuffle_dp} assumes all users follow the protocol. It does not account for malicious users that aim to make the protocol less private. A simple attack is for such users to drop out: for $\gamma \leq 1$, let $\cS\circ\cR^{\gamma n}$ denote the case where only $\gamma n$ out of $n$ users execute $\cR$. Because the behavior of the randomizer may depend on $n$, $\cS\circ\cR^n$ may satisfy a particular level of differential privacy but $\cS\circ\cR^{\gamma n}$ may not\footnote{Note that, with respect to differential privacy, dropping out is ``the worst'' malicious users can do. This is because adding messages from malicious users to those from honest users is a post-processing of $\cS\circ \cR^{\gamma n}$. If $\cS\circ \cR^{\gamma n}$ is already differentially private for the outputs of the $\gamma n$ users alone, then differential privacy's resilience to post-processing (Fact~\ref{fact:post}) ensures that adding other messages does not affect this guarantee. Hence, it is without loss of generality to focus on drop-out attacks.}. Ideally, the privacy guarantee should not suffer too much from a small number of malicious users. This motivates a \emph{robust} variant of shuffle privacy.

\begin{definition} [Robust Shuffle Differential Privacy]
\label{def:robust_shuffle_dp}
	Fix $\gamma\in(0,1]$. A protocol $\cP=(\cR,\cA)$ is \emph{$(\eps,\delta, \gamma)$-robustly shuffle differentially private} if, for all $n\in\N$ and $\gamma' \geq \gamma$, the algorithm $\cS \circ \cR^{\gamma'  n}$ is $(\eps, \delta)$-differentially private. In other words, $\cP$ guarantees $(\eps, \delta)$-shuffle privacy whenever at least a $\gamma$ fraction of users follow the protocol.
\end{definition}

As with generic shuffle differential privacy, we often shorthand this as ``robust shuffle privacy.'' Note that we define robustness with regard to privacy rather than accuracy. A robustly shuffle private protocol promises its users that their privacy will not suffer much from a limited fraction of malicious users. It does not make any guarantees about the accuracy of the analysis. We state our accuracy guarantees under the assumption that all users follow the protocol. In general, we assume $\gamma \in \{1/n, 2/n, \ldots, 1\}$ to avoid ceilings and floors.

We emphasize that robust shuffle privacy is not implied by the generic shuffle privacy of Definition~\ref{def:shuffle_dp}. This is easy to see if we relax the model to allow users to execute different randomizers. In this case, all responsibility to add noise may rest on one user. For example, if each user has a single-bit datum, all users can report their data without noise while one designated user reports multiple randomized bits. If that designated user drops out then the remaining output clearly fails to be differentially private.
(In particular, an attacker who knows $n-1$ of the users' data learns the $n^{th}$ user's data as well.)

At the same time, many existing shuffle protocols are robustly shuffle private. This stems from a common protocol structure: each user contributes some noisy messages (e.g. Bernoulli bits) so that the union of $n$ of these sets (e.g. a binomial distribution) suffices for a target level of differential privacy. A fraction of malicious users worsens the privacy guarantee, but the effect is limited due to the noise contributions of the remaining honest users. For an example of a formal argument, see work by Ghazi, Pagh, and Velingker~\cite{GPV19}. We also give robust shuffle private adaptations of past work on histograms (Claim~\ref{claim:robust_bc}).
\section{Distinct Elements}
\label{sec:distinct}
We begin with the basic problem of counting distinct elements. Without loss of generality, define the data universe to be $\cX = [k]$. For all $\vec{x} \in [k]^n$, let $D(\vec{x})$ denote the number of distinct elements in $\vec{x}$, i.e. $D(\vec{x}) := |\{j \in [k] \mid \exists i \text{ where } x_i=j\}|$.

\begin{definition}[Distinct Elements Problem]
An algorithm $\cM$ solves the \emph{$(\alpha,\beta)$-distinct elements problem} on input length $n$ if for all $\vec{x}\in [k]^n$, $\P{\cM}{ |\cM(\vec{x})- D(\vec{x})| \leq \alpha} \ge 1-\beta$.
\end{definition}

\subsection{Upper Bound for Robust Shuffle Privacy}\label{sec:distinct-upper}
The main idea of our protocol is to reduce the distinct elements problem to computing the $\mor$ function: $D(\vec{x}) = \sum_{i=1}^k ((x_1 = i) \vee \ldots \vee (x_n = i))$. If we can estimate each of these $k$ $\mor$ functions to relatively good accuracy, then we can appropriately de-bias their sum to estimate $D(\vec{x})$. 

The main problem is now to compute $\mor$ under robust shuffle privacy. We start with a basic (and suboptimal) centrally private solution to $\mor$: output a sample from $\tBer{1/2}$ if $\mor(\vec{x}) = 1$, and output a sample from $\tBer{p= 1/2e^\eps}$ otherwise. By concentration across the $k$ instances of $\mor$, this protocol achieves accuracy roughly $O(\sqrt{k} / \eps)$.

The key property of the above solution is that we can simulate it in the shuffle model using a technique from ~\citet{BBGN19b}. The first step is to equate a sample from $\tBer{p}$ with the sum (mod 2) of $n$ samples from some other distribution $\tBer{p'}$; if every user reports messages drawn from $\tBer{p'}$, the sum (mod 2) is exactly $\tBer{p}$. But if there is any user who reports $\tBer{1/2}$, the sum is drawn from $\tBer{1/2}$. So we can exactly simulate the basic centrally private solution above once we can implement modular arithmetic in the shuffle model. Ishai, Kushilevitz, Ostrovsky, and Sahai \cite{IKOS06} provide a solution using additive shares:
\begin{theorem}[\cite{IKOS06,BBGN19b}]
\label{thm:ikos}
There exists a shuffle protocol $\cP_\mmod=(\cR_\mmod,\cA_\mmod)$ that receives input $\vec{x}\in\zo^n$ along with security parameter $\sigma > 0$ and outputs $\sum_{i=1}^n x_i$ mod 2 if all users are honest. If only $\gamma n \geq 2$ users are honest, let $\vec{h}$ denote the vector of their values. Then $\tv{(\cS\circ\cR^{\gamma n}_\mmod)(\vec{h})}{(\cS\circ\cR^{\gamma n}_\mmod)(\vec{w})} < 2^{-\sigma}$, where $\vec{w} := (\sum h_i \mod 2, 0, \dots, 0)$. Each honest user sends $O(\sigma + \log n)$ one-bit messages.
\end{theorem}
Our distinct elements protocol $\cP_\de = (\cR_\de, \cA_\de)$ will choose the security parameter $\sigma$ and invoke $\cR_\mmod,\cA_\mmod$ as subroutines. The pseudocode appears in Algorithms~\ref{alg:de-upper-randomizer} and~\ref{alg:de-upper-analyzer}.

\begin{algorithm}[h]
\caption{Randomizer $\cR_\de$}
\label{alg:de-upper-randomizer}
\KwIn{user data $x \in [k]$; number of users $n$, privacy~parameters~$\eps > 0$~and~$\delta \in (0,1)$}
\KwOut{message vector $\vec{y} \in ([k] \times \{0,1\})^*$}
\BlankLine
\BlankLine
Initialize output messages $\vec{y} \gets \emptyset$

Set $p' \gets \frac{1 - (1 - e^{-\eps})^{1/n}}{2}$

Set $\sigma \gets \log\left(\frac{e^\eps+1}{\delta}\right)$

%Set $m \gets 7 + \lceil 2\sigma + 2\log(n-1) \rceil$

\For{domain element $j \in [k]$}{
	\If{$x = j$}{Draw $u^{(j)} \sim \tBer{1/2}$}
	\Else{ Draw $u^{(j)} \sim \tBer{p'}$ }
	
	%Sample $\ell \in \{0,1\}^m$ uniformly from $\{v \in \{0,1\}^m \mid \sum_{t=1}^m v_t \text{ mod } 2 = u^{(j)}\}$
	Obtain messages $\ell \in \{0,1\}^m \gets \cR_\mmod(u^{(j)})$
	
	\For{$t \in [m]$}{
		Append $(j,\ell_t)$ to $\vec{y}$
	}
}

\Return $\vec{y}$
\end{algorithm}

\begin{algorithm}[H]
\caption{Analyzer $\cA_\de$}
\label{alg:de-upper-analyzer}
\KwIn{message vector $\vec{y} \in ([k] \times \zo)^*$; number of users $n \in \N$, privacy~parameters~$\eps > 0$~and~$\delta \in (0,1)$}
\KwOut{$z \in \R$}
\BlankLine
\BlankLine
Set $\sigma \gets \log\left(\frac{e^\eps+1}{\delta}\right)$

\For{domain element $j \in [k]$}{
	Initialize $\vec{y}^{(j)} \gets \emptyset$
	
	\For{$(j,\ell_t) \in \vec{y}$}{
		Append $\ell_t$ to $\vec{y}^{(j)}$
	}
	
	$C_j \gets \cA_\mmod(\vec{y}^{(j)})$ %$C_j \gets \sum_{t=1}^{|\vec{y}^{(j)}|} \vec{y}^{(j)}_i$ mod 2
}

$C \gets \sum_{j=1}^{k} C_j$

\Return $z \gets \frac{2Ce^\eps - k}{e^\eps-1}$
\end{algorithm}

\begin{theorem}
\label{thm:de_upper}
Given $\eps > 0$, $\gamma \in (0,1]$, and $\beta, \delta \in (0,1)$, the protocol $\cP_\mde = (\cR_\mde, \cA_\mde)$
\begin{enumerate}[label=\Roman*.]
	\item is $\left(2\eps(\gamma), \tfrac{4\delta}{\gamma}, \gamma\right)$-robustly shuffle private, where $\eps(\gamma) \le \eps + \ln\left(\tfrac{1}{\gamma}\right)$ and if $\eps \le \ln(2)$, then $\eps(\gamma) \le 2\tfrac{\eps^\gamma}{\gamma}$;
	\item solves the $(\alpha,\beta)$-distinct elements problem for
	\begin{align*}
	\alpha = \frac{e^\eps}{e^\eps-1} \cdot \sqrt{2k \ln(2/\beta)} = O\left(\max\left(1, \tfrac{1}{\eps}\right)\cdot\sqrt{k\log(1/\beta)}\right);
	\end{align*}
	\item requires each user to communicate at most $O\left(k\log\left(\frac{n(e^\epsilon+1)}{\delta}\right)\right)$ messages of length $O(\log(k))$.
\end{enumerate}
\end{theorem}
\begin{proof}
\underline{Privacy (I)}:  We will show that $(\cS\circ \cR^{\gamma n}_\de)$ is differentially private. Consider neighboring datasets $\vec{x} \sim \vec{x}\,' \in [k]^{\gamma n}$ where $x_i = a \neq b = x_i'$. The output distributions of $\cR_\de(a)$ and $\cR_\de(b)$ differ only in the messages labeled by $a$ and $b$, $\vec{y}^{(a)}$ and $\vec{y}^{(b)}$. It follows that to prove privacy we need only analyze the distributions of $\vec{y}^{(a)}, \vec{y}^{(b)}$.

To do so, let $\vec{u}^{(a)} = (u^{(a)}_1,\dots, u^{(a)}_n)$ be the binary vector where each $u^{(a)}_i \sim \tBer{1/2}$ when $x_i = a$ and $u^{(a)}_i \sim \tBer{p'}$ otherwise, where we defined $p' = \tfrac{1 - (1 - e^{-\eps})^{1/n}}{2}$ in the pseudocode for $\cR_\de$. We define $\vec{u}^{(a)}\,'$ similarly for $\vec{x}\,'$.
%Let $\cR_{\de,*}$ denote the randomizer that takes a bit $u^{(j)}$ as input and simply executes line 10 of $\cR_\de$ and outputs the result.
It will suffice to prove that, for any $T \subseteq \zo^*$,
\begin{align*}
\P{}{(\cS\circ \cR^{\gamma n}_{\mmod})(\vec{u}^{(a)}) \in T } &\leq e^{\eps(\gamma)} \cdot \P{}{(\cS\circ \cR^{\gamma n}_{\mmod})(\vec{u}^{(a)}\,' ) \in T} + \frac{2\delta}{\gamma} \\
\P{}{(\cS\circ \cR^{\gamma n}_{\mmod})(\vec{u}^{(a)}\,') \in T } &\leq e^{\eps(\gamma)} \cdot \P{}{(\cS\circ \cR^{\gamma n}_{\mmod})(\vec{u}^{(a)} ) \in T} + \frac{2\delta}{\gamma}.
\end{align*}
This is because the inequalities above guarantee $\left(\eps(\gamma), \tfrac{2\delta}{\gamma}\right)$-privacy for the view of $\vec{y}^{(a)}$. Identical arguments hold for $\vec{u}^{(b)}$ and $\vec{u}^{(b)}\,'$, so our privacy guarantee follows from composition.

These results rely on the following lemma.
\begin{lemma}[Lemma 1.2 \cite{BBGN19b}]
\label{lem:bbgn_close}
Let $\cM, \cM'$ be algorithms such that, for every $\vec{u}$, $\tv{\cM(\vec{u})}{\cM'(\vec{u})} \leq \Delta$. If $\cM$ is $\eps$-differentially private then $\cM'$ is $(\eps,(e^\eps + 1) \Delta)$-differentially private.
\end{lemma}
Our protocol is built on top of the modular arithmetic protocol $\cP_\mmod$ of Theorem \ref{thm:ikos},
%of Balle, Bell, Gasc\'on, and Nissim~\cite{BBGN19b}, which itself adapts a protocol from Ishai, Kushilevitz, Ostrovsky, and Sahai \cite{IKOS06},
so we immediately satisfy the distance condition. To be precise, define $\cM(\vec{u})$ to be the algorithm that takes input $\vec{y}$ and outputs $(\cS\circ \cR^{\gamma n}_{\mmod})((\sum u_t \textrm{ mod } 2, 0, \dots, 0) )$. We are guaranteed that $(\cS\circ \cR^{\gamma n}_{\mmod})(\vec{u})$ is within total variation distance $2^{-\sigma}$ of $\cM(\vec{u})$.

Our new goal is to prove that $\cM$ is $\eps'$-differentially private where
$$\eps' := \ln\left(\frac{1}{1 - (1-e^{-\eps})^{\gamma}}\right).$$
Once we do so, we can use Lemma~\ref{lem:bbgn_close} to conclude that $(\cS\circ \cR^{\gamma n}_{\de,*})$ is $(\eps', \delta' )$-differentially private, where $\delta' := (e^{\eps'} + 1) \cdot 2^{-\sigma} = \tfrac{e^{\eps'} + 1}{e^\eps + 1} \cdot \delta$. Hence, we need to show that the following holds for both $z=1$ and $z=0$:
\begin{equation}
\label{eq:zs_ineq}
e^{-\eps'} \leq  \frac{\P{}{\sum u^{(a)}_i \textrm{ mod } 2 = z}}{\P{}{\sum u^{(a)}_i\,' \textrm{ mod } 2 = z}} \leq e^{\eps'}
\end{equation}

$x_i = a$, so $\sum_{t=1}^{\gamma n} u^{(a)}_t \textrm{ mod } 2$ is distributed as $\tBer{1/2}$. If $a \in \vec{x}\,'$ then $\sum_{i=1}^{\gamma n} u^{(a)}_i\,' \textrm{ mod } 2$ is distributed as $\tBer{1/2}$ as well, so the ratio of probabilities is exactly 1. Otherwise, $\sum_{i=1}^{\gamma n} u^{(a)}_i\,' \textrm{ mod } 2$ is a sum of i.i.d.\ draws from $\tBer{p'}$. We argue that this is distributed as $\tBer{1/2e^{\eps'}}$ for $\eps' = \ln\left(\tfrac{1}{1 - (1 - e^{-\eps})^{\gamma}}\right)$ using the following lemma, proven in Appendix \ref{subsec:distinct-lemmas}:

\begin{lemma}
\label{lem:distinct-ber}
	Let $n \in \N$, $\gamma \in (0,1]$, and $p \in [0, 1/2]$. Define $p' = \tfrac{1-(1-2p)^{1/n}}{2}$. Then given i.i.d.\ $X_1, \ldots, X_{\gamma n} \sim \tBer{p'}$, $X = \sum_{i=1}^{\gamma n} X_i \textrm{ mod } 2$ is identically distributed with \[\Ber{\frac{1-(1-2p)^\gamma}{2}}.\]
\end{lemma}

Using Lemma~\ref{lem:distinct-ber} with our $p'$ gives $p = \tfrac{1}{2e^\eps}$, so
$$\sum_{i=1}^{\gamma n} u^{(a)}_i\,' \textrm{ mod }  2 = \Ber{\frac{1 - (1 - e^{-\eps})^\gamma}{2}} = \Ber{\frac{1}{2e^{\eps'}}}$$
and we return to Equation~\ref{eq:zs_ineq} to find
$\tfrac{\P{}{\Ber{\half} = 1}}{\P{}{\Ber{\tfrac{1}{2e^{\eps'}}} = 1}} = e^{\eps'}$ and $\tfrac{\P{}{\Ber{\half} = 0}}{\P{}{\Ber{\tfrac{1}{2e^{\eps'}}} = 0}} = \tfrac{1}{2-e^{-\eps'}}$. Because $\eps'>0$, we have $e^{\eps'}+ e^{-\eps'} > 2$ and therefore $2-e^{-\eps'}<e^{\eps'}$. Thus we have proven the inequality in Equation~\ref{eq:zs_ineq}.

All that is left is to prove $\eps' \leq \eps(\gamma)$ and $\delta' \leq 2\delta/\gamma$. To bound $\eps'$, we split into cases based on $\eps$.

\vb{case 1 holds for all $\eps$, so only need to bound $\delta$ once}
\underline{Case 1}: $\eps > \ln(2)$. Then, as defined in the theorem statement, $\eps(\gamma) = \eps + \ln(1/\gamma)$. We use the following variant of Bernoulli's inequality: for $x \geq -1$ and $r \in [0,1]$, $(1+x)^r \leq 1+xr$. Using $r=\gamma$ and $x = -e^{-\eps}$, we have $(1-e^{-\eps})^\gamma \leq 1-\gamma e^{-\eps}$ and therefore
$$ \eps' = \ln\left(\frac{1}{1 - (1-e^{-\eps})^{\gamma}}\right) \leq \ln \left(\frac{1}{1-(1-\gamma e^{-\eps})} \right) = \eps + \ln\left(\frac{1}{\gamma}\right).$$

\underline{Case 2}: $\eps \leq \ln(2)$, in which case we wish to show $\eps' \leq \frac{2\eps^\gamma}{\gamma}$. The inequality $1+x \leq e^x$ implies the following:
$$\eps' = \ln\left(\frac{1}{1 - (1 - e^{-\eps})^{\gamma}}\right) = \ln\left(\frac{e^{\eps\gamma}}{e^{\eps\gamma} - (e^{\eps} - 1)^{\gamma}}\right) = \ln \left( 1 + \frac{(e^{\eps} - 1)^{\gamma} }{e^{\eps\gamma} - (e^{\eps} - 1)^{\gamma}} \right) \leq \frac{(e^{\eps} - 1)^{\gamma} }{e^{\eps\gamma} - (e^{\eps} - 1)^{\gamma}}$$
Given that $\eps \leq \ln(2)$, we have $e^\eps-1 \leq \tfrac{e^\eps}{2}$. In turn, $(e^\eps-1)^\gamma \leq \tfrac{e^{\eps\gamma}}{2^\gamma}$. Thus
\begin{align*}
\frac{(e^{\eps} - 1)^{\gamma} }{e^{\eps\gamma} - (e^{\eps} - 1)^{\gamma}} &\leq\ \frac{(e^\eps-1)^\gamma}{e^{\eps \gamma} - \tfrac{e^{\eps \gamma}}{2^\gamma}} \\
&= \frac{2^\gamma}{2^\gamma - 1} \cdot \left(\frac{e^\eps-1}{e^\eps}\right)^\gamma\\
	&\leq \frac{2^{\gamma}}{2^\gamma - 1}\cdot \eps^\gamma
\end{align*}
because $\tfrac{e^\eps-1}{e^\eps} = 1 - e^{-\eps} \leq \eps$. We use the following lemma, also proven in Appendix \ref{subsec:distinct-lemmas}:
\begin{lemma}
\label{lem:distinct-gamma}
	For $\gamma \in (0,1]$, $\tfrac{2^\gamma}{2^\gamma-1} \leq \tfrac{2}{\gamma}$.
\end{lemma}

Thus $\eps' \leq \frac{2\eps^\gamma}{\gamma}$.

To bound $\delta'$, we note that $\eps' \leq \eps +\ln(1/\gamma)$ holds for all $\eps$. We substitute into our definition of $\delta'$ to get
$$\delta' = \delta \cdot \frac{e^{\eps'}+1}{e^\eps+1} \leq \delta \cdot \frac{\tfrac{e^\eps}{\gamma}+1}{e^\eps+1} \leq \frac{\delta}{\gamma}$$
since $\gamma \in (0,1]$.

% We now prove $\delta' \leq 2\delta/\gamma$ using
% \begin{align*}
% \frac{e^{\eps'}+1}{e^\eps+1} =&\  \frac{2 + \frac{(e^{\eps} - 1)^{\gamma} }{e^{\eps\gamma} - (e^{\eps} - 1)^{\gamma}} }{e^\eps + 1} \\
	% <&\ 1 + \half \cdot \frac{(e^{\eps} - 1)^{\gamma} }{e^{\eps\gamma} - (e^{\eps} - 1)^{\gamma}} \tag{$\eps > 0$}\\
	% \leq&\ 1 + \frac{\eps^\gamma}{\gamma}\\
	% \leq&\ 2/\gamma. \tag{$\eps,\gamma < 1$}
% \end{align*}

\underline{Accuracy (II)}: If all $n$ users follow the protocol, then $\sum_{i=1}^n u^{(j)}_i \textrm{ mod } 2$ is distributed as $\tBer{1/2}$ when there is some $x_i=j$. Otherwise, Lemma \ref{lem:distinct-ber} implies the distribution is $\tBer{1/2e^\eps}$. Thus $C = \sum_{j=1}^k C_j$ has expectation $\E{}{C} = \frac{D(\vec{x})}{2} + \frac{k-D(\vec{x})}{2e^\eps}$. In turn, the output of $\cP_\mde$ has expectation $\E{}{\tfrac{2Ce^\eps -k}{e^\eps-1}} = D(\vec{x})$. A Hoeffding bound implies that
$$\P{}{|\cP_\mde(\vec{x}) - D(\vec{x})| > \frac{e^\eps}{e^\eps-1}\sqrt{2k\ln \frac{2}{\beta} }} \leq \beta$$
where the probability is over users' randomizers.

\underline{Communication (III)}: For each domain element, Theorem \ref{thm:ikos} ensures there are $O(\log(n) + \sigma)$ messages, each of which is one labeled bit. Since each label is $ \in [k]$, each user sends $O(k[\log(n) + \sigma])$ messages of length $O(\log(k))$. Substituting in $\sigma = \log\left(\tfrac{e^\eps+1}{\delta}\right)$ yields the claim.
\end{proof}

We focus on the setting where $n \geq k$, as this is the setting for our lower bound in the next section. For completeness, we also give an $O(n^{2/3})$ guarantee for the small-$n$ setting. At a high level, this modified protocol simply hashes the initial domain $[k]$ to a smaller domain of size $O(n^{4/3})$ and then runs the protocol given above for this new domain. Details appear in Appendix~\ref{subsec:app_distinct_2}.

\subsection{Lower Bound for Robust Shuffle Privacy}
We now show that this dependence on $k$ is tight for the setting where $n = \Omega(k)$. To do so, we give a way to transform a robustly shuffle private protocol into a pan-private one (Algorithm~\ref{alg:de-transformation}) and then invoke a lower bound for pan-private distinct elements~\cite{MMNW11}.

Our transformation is simple: the pan-private algorithm uses the shuffle protocol to maintain a set of shuffle protocol messages as its internal state. More concretely, the pan-private algorithm initializes its internal state using $n/3$ draws from the protocol randomizer $\cR(1)$, processes the stream $\vec{x}$ by adding $\cR(x_1), \ldots, \cR(x_{n/3})$ to its collection of messages, adds another $n/3$ draws from $\cR(1)$ to its internal state after the stream, and finally applies the protocol analyzer $\cA$ to this final internal state to produce output. Pan-privacy follows from the original protocol's robust shuffle privacy combined with our incorporation of ``dummy'' messages into the state. By the original protocol's accuracy guarantee, these dummy messages -- all generated from a single element-- increase final error by at most 1.

We remark that this construction assumes $n$ is a multiple of 3, but this constraint can be removed by using $\lceil n/3\rceil$ and $\lfloor n/3 \rfloor$ where appropriate. We avoid this technicality for sake of clarity.

\begin{algorithm}

\caption{$\cQ_\cP$, an online algorithm for distinct elements}

\label{alg:de-transformation}

\KwIn{Data stream $\vec{x}\in[k]^{n/3}$; a shuffle protocol $\cP=(\cR,\cA)$ for distinct elements}

\KwOut{An integer in $[k]$}

Let vector $\ones \gets (1, \ldots, 1) \in \mathbb{N}^{n/3}$

Initialize internal state $I_0 \gets (\cS\circ\cR^{n/3})(\ones)$

\For{$i\in[n/3]$}{
    Set $I_i \gets \cS(I_{i-1}, \cR(x_i))$
}

Set final state $\vec{y} \gets \cS(I_{n/3}, \cR^{n/3}(\ones))$

\Return{$\cA(\vec{y})$}

\end{algorithm}

\begin{lemma}
\label{lem:de-transformation}
Suppose there exists a protocol $\cP=(\cR,\cA)$ that is $(\eps, \delta, 1/3)$-robustly shuffle private and solves the $(\alpha,\beta)$-distinct elements problem on input length $n$. Then $\cQ_\cP$ is an $(\eps, \delta)$-pan-private algorithm that solves the $(\alpha+1,\beta)$-distinct elements problem on input length $n/3$.
\end{lemma}
\begin{proof}
\underline{Privacy}: The main idea of the proof is that, by the robust shuffle privacy of $\cP$, the first draw from $(\cS\circ\cR^{n/3})(\ones)$ ensures privacy of the internal state view, and the second draw ensures privacy for the output view. For clarity, we write this out explicitly below.

By construction, messages added to the state are sampled independently of the current state.
Now, fix intrusion time $t$.
Let $I_{>t}$ be the shuffled collection of messages generated after the intrusion at time $t$.
I.e. $I_{>t} \sim \cS\big(\cR(x_{t+1}), \ldots, \cR(x_{n/3}), \cR^{n/3}(\vec{1})\big)$.
Notice that
\begin{align*}
(I_t, \cQ_\cP(\vec{x})) \sim \big(I_t, (\cA\circ\cS)(I_t, I_{>t})\big).
\end{align*}
So $(I_t, \cQ_\cP(\vec{x}))$ is a post-processing of $(I_t, I_{>t})$.
Thus, to show $\cQ_\cP$ is $(\eps, \delta)$-pan-private, it suffices to show that the algorithm outputting $(I_t, I_{>t})$ is $(\eps, \delta)$-differentially private.

Consider a neighboring data stream $\vec{x}\,'$.
Define $I'_t$ and $I'_{>t}$ as the internal state at time $t$ and the shuffled collection of messages generated after time $t$ for $\cQ_\cP$ running with input $\vec{x}\,'$.

Let $j$ be the index on which $\vec{x}$ and $\vec{x}\,'$ differ.
If $j > t$, then $I_t$ has the same distribution as $ I_t'$.
Additionally, $I_{>t}$ is drawn from $(\cS\circ\cR^{2n/3-t})(x_{t+1},\dots,x_n,1,\dots,1)$ and $I'_{>t}$ is drawn from $(\cS\circ\cR^{2n/3-t})(x'_{t+1},\dots,x'_n,1,\dots,1)$. The inputs differ on one index and there are at least $n/3$ executions of $\cR$. %and $I_{>t}'$ are distributed identically to the outputs of the shuffler for protocol $\cP$ on neighboring datasets with at least $n/3$ users since the messages for $n/3$ 1s were added.
Therefore, $I_{>t}$ and $I_{>t}'$ are $(\eps, \delta)$-indistinguishable, i.e. $\Pr[I_{>t} \in T] \le e^\eps \cdot \Pr[I'_{>t} \in T]  + \delta$ for all $T$, which then implies the joint distributions $(I_t, I_{>t})$ and $(I_t', I_{>t}')$ are $(\eps, \delta)$-indistinguishable.
A similar argument holds for the case when $j \le t$.
Therefore, $\cQ_\cP$ is $(\eps, \delta)$-pan-private.

\underline{Accuracy}: Consider the vector $\vec{w}=(1,\dots,1, x_1,\dots,x_{n/3},1,\dots,1)\in [k]^n$. By the accuracy guarantee of the original shuffle protocol $\cP$, we have $\P{}{|\cP(\vec{w})-D(\vec{w})|>\alpha }<\beta$. By the construction of $\cQ_{\cP}$, $\cQ_{\cP}(\vec{x})$ is identically distributed to $\cP(\vec{w})$. Combining the triangle inequality and $|D(\vec{x})-D(\vec{w})| \leq 1$, we conclude $\cQ_{\cP}$ solves the $(\alpha+1,\beta)$-distinct elements problem .
\end{proof}

We now recall the pan-private lower bound for distinct elements.~\citet{MMNW11} stated their result for pure user-level pan-privacy. However, the same argument works for record-level privacy. The proof also concludes by recovering, for $\omega(1)$ elements, whether or not those elements appeared in the stream. This immediately yields the following approximate record-level result:

\begin{lemma}[Implicit in Corollary 3~\cite{MMNW11}]
\label{lem:de-lower-bound-pan}
If $(\eps,\delta)$-pan-private $\cQ$ solves the $(\alpha,\beta)$-distinct elements problem on input length $n$ for $\alpha = o(\sqrt{k})$, $\beta < 0.1$, and $n \geq k$, then $\eps=\omega(1)$ or $\delta=\omega\left(\tfrac{1}{n}\right)$.
\end{lemma}

%de-lower-bound-pan2-victor

% OLD VERSION
\iftrue%de-lower-bound-pan2
In fact, we can strengthen Lemma~\ref{lem:de-lower-bound-pan} to incorporate $\eps$. Throughout, we ignore ceilings and floors for neatness.

\vb{old version}
\begin{lemma}
\label{lem:de-lower-bound-pan-2}
	Let $\eps \leq 1$, $\delta = O\left(\tfrac{\eps}{n}\right)$, and $\beta < 0.1$. Suppose there exists $(\eps,\delta)$-pan-private $\cQ$ that solves the $(\alpha, \beta)$-distinct elements problem on input length $n \geq \tfrac{k}{\eps}$ for domain $[k]$. Then $\alpha = \Omega\left(\sqrt{\tfrac{k}{\eps}}\right)$.
\end{lemma}
\begin{proof}
	We first show that $\cQ$ yields a $\left(1,\tfrac{\delta}{\eps}\right)$-pan-private $\cQ'$ that solves the $(\alpha \eps, \beta)$-distinct elements problem on input length $n' \geq k\eps$ for domain $[k\eps]$.
	
	At a high level, $\cQ'$ transforms distinct elements on $[k\eps]$ into distinct elements over a larger domain $[k]$ and uses $\cQ$. Concretely, $\cQ'$ initializes $\cQ$ and then, for each received element $j_2 \in [k\eps]$, creates $\tfrac{1}{\eps}$ copies $(j_2, 1), \ldots (j_2, 1/\eps)$ and passes them to $\cQ$. The cost is that, by composition across the $\tfrac{1}{\eps}$ copies passed to $\cQ$, each element from $[k\eps]$ is now only guaranteed $\left(1, \tfrac{\delta}{\eps}\right)$-pan-privacy in the state maintained by $\cQ$. The stream passed to $\cQ$ has length $\tfrac{n'}{\eps} \geq k$, so with probability at least $1-\beta$, $\cQ$ outputs an $\alpha$-accurate estimate of the number of distinct elements for a domain of size $[k]$; by our transformation, $\cQ'$ can multiply this by $\eps$ to get an $\alpha \eps$-accurate estimate of the number of distinct elements in its stream from $[k\eps]$.
	
	We now apply Lemma~\ref{lem:de-lower-bound-pan} to $\cQ'$. To check that the required conditions hold, $\delta = O\left(\tfrac{\eps}{n}\right)$ implies $\tfrac{\delta}{\eps} = O\left(\tfrac{1}{n}\right)$, the input length $n' \geq k\eps$, and $\beta < 1/10$.  Thus $\alpha\eps = \Omega(\sqrt{k\eps})$, and we rearrange into $\alpha = \Omega\left(\sqrt{\tfrac{k}{\eps}}\right)$.
\end{proof}
\fi%de-lower-bound-pan2

We now combine Lemma~\ref{lem:de-transformation} and Lemma~\ref{lem:de-lower-bound-pan-2} to get the following robust shuffle private lower bound.

\begin{theorem}
\label{thm:de-lower-bound}
	Let $\eps \leq 1$, $\delta = O\left(\tfrac{\eps}{n}\right)$, $\beta < 0.1$. If $\cP$ is $(\eps,\delta,1/3)$-robustly shuffle private and solves the $(\alpha, \beta)$-distinct elements problem on input length $n \geq \tfrac{3k}{\eps}$, then $\alpha = \Omega\left(\sqrt{\tfrac{k}{\eps}}\right)$.
\end{theorem}

In contrast, the trivial $\eps$-centrally private solution of computing the number of distinct elements and adding geometric noise is achieves error $O(1/\eps)$ for any $\eps$ and $n$.
\section{Uniformity Testing}
\label{sec:uniform}
We now move to the second main problem of this paper, uniformity testing. A uniformity tester uses i.i.d.\ sample access to an unknown distribution over a domain $[k]$ to distinguish the cases where the distribution is uniform or far from uniform. 

\begin{definition}[Uniformity Testing]
An algorithm $\cM$ solves $\alpha$-\emph{uniformity testing} with sample complexity $m$ when:
\begin{itemize}
    \item If $\vec{x}\sim \bU^m$, then $\P{}{\cM(\vec{x}) = \textrm{``uniform''}} \geq 2/3$, and
    \item If $\vec{x}\sim \bD^m$ where $\tv{\bD}{\bU} > \alpha$, then $\P{}{\cM(\vec{x}) = \textrm{``not uniform''}} \geq 2/3$
\end{itemize}
where the probabilities are taken over the randomness of $\cM$ and $\vec{x}$.
\end{definition}

Note that achieving an overall $2/3$ success probability is essentially equivalent to achieving $\Omega(1)$ separation between the probabilities of outputting ``uniform'' given uniform and non-uniform samples. This is because any such $\Delta$ separation can be amplified using $O\left(\tfrac{1}{\Delta^2}\right)$ repetitions. For this reason, we generally focus on achieving any such constant separation.

Our algorithms will also often rely on \emph{Poissonization} and use not $m$ samples but $n \sim \Pois(m)$. Doing so ensures that sampled counts of different elements are independent over the randomness of drawing $n$, which will be useful in their analysis. Fortunately, $\Pois(m)$ concentrates around $m$~\cite{C17}. We can therefore guarantee $O(m)$ samples at the cost of a constant decrease in success probability. Because we generally focus on constant separations, we typically elide the distinction between ``sample complexity $m$'' and ``sample complexity distributed as $\Pois(m)$''.

Throughout, in the shuffle model we assume that users receive i.i.d.\ samples from $\bD$, one sample per user. In the pan-private model, we assume that the stream consists of i.i.d.\ samples from $\bD$.

\subsection{Upper Bound for Robust Shuffle Privacy}
\label{subsec:uniform_upper}

In this section, we give a robustly shuffle private uniformity tester. At a high level, our protocol imitates the pan-private uniformity tester of~\citet{AJM19} (which itself imitates a centrally private uniformity tester suggested by~\citet{CDK17}). This tester maintains $k$ sample counts, one for each element, and compares a $\chi^2$-style statistic of the counts to a threshold to determine its decision. To ensure privacy, the algorithm adds Laplace noise to each count before computing the statistic.

Our protocol is similar, but the shuffle model introduces a complication: we cannot privately count the frequency of a universe element, but instead rely on private communications from users. More concretely, let $c_j(\vec x)$ denote the true count of $j \in [k]$ in $\vec{x}$. Users will (implicitly) execute a private binary sum protocol to count each $j$. The analyzer obtains a  vector of estimates $(\tilde c_{1}(\vec{x}), \ldots, \tilde c_{k}(\vec{x}))$ and uses them to compute the test statistic used by~\citet{AJM19} (Section \ref{sec:prelim-ut}). We then apply the binning trick from~\citet{AJM19} -- roughly, maintaining coarser counts for random groups of elements rather than every element separately -- to obtain our final uniformity tester with sample complexity $\tilde O(k^{2/3})$ (Section \ref{sec:final-ut}).

\subsubsection{Preliminary Uniformity Tester}
\label{sec:prelim-ut}
We give a preliminary robustly shuffle private uniformity tester $\cP_\ut = (\cR_\ut,\cA_\ut)$ (Algorithms~\ref{alg:ut-r} and~\ref{alg:ut-a}). This tester first compiles robust shuffle private estimates of the sample counts for each element in $[k]$. It then uses these counts to compute a statistic $Z'$ that is, roughly, large when the underlying distribution is sufficiently non-uniform and small otherwise. The resulting uniformity obtains sample complexity scaling with $k^{3/4} \ln^{1/2}(k)$. In the next section, we use this initial tester as a black box to obtain a tester that improves $k^{3/4}$ to $k^{2/3}$.

\begin{algorithm}[h]
\caption{$\cR_\ut$, a randomizer for private uniformity testing}
\label{alg:ut-r}

\KwIn{User data point $x\in[k]$; parameters $\lambda,n \in \N$}
\KwOut{A message vector $\vec{y} \in ([k] \times \zo)^*$}

\For{$j\in[k]$}{
    
    \If{$x = j$}{Set message $\vec{y} \gets (j,1)$}
    \Else{Set data point message $\vec{y} \gets (j,0)$}
	
	Draw number of noisy messages $s_j \sim \Pois(\lambda/n)$
    
    \For{$t \in [s_j]$}{
        $b_{j,t} \sim \Ber{1/2}$
    
        Append $(j,b_{j,t})$ to $\vec{y}$
    }
}

\Return{$\vec{y}$}
\end{algorithm}

\begin{algorithm}[H]
\caption{$\cA_\ut$, an analyzer for private uniformity testing}
\label{alg:ut-a}

\KwIn{A message vector $\vec{y}\in([k]\times \zo)^*$; parameters $\tau\in\R$, $\lambda,m,n \in \N, \alpha \in (0,1)$}
\KwOut{A string in $\{ \textrm{``uniform''}, \textrm{``not uniform''} \}$}

\For{$j\in[k]$}{
    
    Calculate noise scale $\ell_j \gets -n + |\{y_i \mid y_i = (j,0) \text{ or } y_i = (j,1)\}|$
    
    De-bias estimate of count of $j$ as $c_j(\vec{y}) \gets -\ell_j/2 + |\{y_i \mid y_i = (j,1)\}|$
}

Compute statistic $Z' \gets \frac{k}{m} \sum_{j=1}^k (c_j(\vec{y}) -m/k)^2-c_j(\vec{y})$

\Return{\emph{``not uniform''} if $Z'>\tau$ otherwise \emph{``uniform''} }

\end{algorithm}

\begin{theorem}
\label{thm:ut-upper-bound-prelim}
	Let $\gamma \in (0,1]$, $\eps > 0$, and $\alpha, \delta \in (0,1)$. There exists parameters $\lambda \in \N$ and $\tau\in \R$ such that the protocol $\cP_\ut = (\cR_\ut,\cA_\ut)$ is $(2\eps, 8\delta^\gamma,\gamma)$-robustly shuffle private. For $\eps = O(1)$ and $\delta = o(1)$, $\cP_\ut$ solves $\alpha$-uniformity testing with sample complexity
	$$m = O\left(\frac{k^{3/4}}{\alpha \eps} \ln^{1/2}\left(\frac{k}{\delta}\right) + \frac{k^{2/3}}{\alpha^{4/3} \eps^{2/3}} \ln^{1/3}\left(\frac{k}{\delta}\right) + \frac{k^{1/2}}{\alpha^2}  \right).$$
\end{theorem}
\begin{proof}
\underline{Privacy}: We shall prove that $(\cS\circ\cR_\ut^{\gamma n})$ is $(2\eps, 8\delta^\gamma)$ private. Let $\vec{x} \sim \vec{x}\,'$ be neighboring datasets of size $\gamma n$. Without loss of generality, suppose $x_1 = j \neq x_1' = j'$. The distributions of the messages containing $j'' \neq j, j'$ are identical for $\vec{x}$ and $\vec{x'}$, so it suffices to restrict our analysis to messages containing either $j$ or $j'$. To prove our guarantee, by composition it is enough to show that the pool of messages containing $j$ and the pool of messages containing $j'$ are each $(\eps,4\delta^\gamma)$-shuffle private.

We prove this claim for $j$; the $j'$ case is identical. Recall that both the binomial and Poisson distribution are closed under summation:
$$\sum_{i=1}^r \Bin(a_i, p) = \Bin\left(\sum_{i=1}^r a_i, p\right) \text{ and } \sum_{i=1}^r \Pois(\lambda_i) = \Pois\left(\sum_{i=1}^r \lambda_i\right).$$
Recall also that the count of noisy messages $(j,0),(j,1)$ generated by any honest user is distributed as $\Pois(\lambda/n)$. Thus the number of times $(j,1)$ occurs in $(\cS\circ \cR^{\gamma n})(\vec{x})$ is distributed as
$$c_j(\vec{x}) + \sum_{i=1}^{\gamma n} \Bin\left(\Pois\left(\frac{\lambda}{n}\right),1/2\right) = c_j(\vec{x}) + \Bin\left(\Pois\left(\gamma \lambda\right), 1/2\right).$$

Let $\lambda = 40\cdot \left(\tfrac{e^\eps+1}{e^\eps-1}\right)^2\ln\left(\tfrac{2k}{\delta}\right)$. We rely on the following concentration lemma:
\begin{lemma}[Theorem 1~\cite{C17}]
\label{lem:poi_con}
	For $X \sim \Pois(\lambda)$ and $t>0$,
$$\P{}{|X - \lambda| \geq t} \leq 2e^{-\tfrac{t^2}{2(\lambda+t)}}.$$
\end{lemma}
Specifically, the lemma implies that $\P{}{\Pois(\gamma \lambda) < \gamma \lambda/2} < 2\exp(-\gamma\lambda/12)$. In our case, the right-hand side is bounded by $2\exp(-\ln (2/\delta)^\gamma) < 2\delta^\gamma$. By Lemma~\ref{lem:b_noise}, adding $\Bin(\gamma\lambda/2,1/2) = \Bin(20\cdot \left(\tfrac{e^\eps+1}{e^\eps-1}\right)^2 \ln(\tfrac{2k}{\delta})^\gamma , 1/2)$ suffices for $(\eps,2\delta^\gamma)$-differential privacy. After a union bound, we get that the pool of messages containing $j$ is $(\eps,4\delta^\gamma)$-shuffle private.

\underline{Sample Complexity}: Let $E_j = c_j(\vec{y}) - c_j(\vec{x})$. We rewrite $Z'$ in terms of $E_j$:
\begin{align*}
Z' &= \frac{k}{m} \sum_{j=1}^k \left[\left(c_j(\vec{y}) - \frac{m}{k}\right)^2-c_j(\vec{y})\right] \\
    &= \frac{k}{m} \sum_{j=1}^k \left[\left(c_j(\vec{x})+E_j-\frac{m}{k}\right)^2 - (c_j(\vec{x})+E_j)\right] \tag{By definition}\\
    &= \underbrace{\frac{k}{m} \sum_{j=1}^k \left[\left(c_j(\vec{x})-\frac{m}{k}\right)^2 - c_j(\vec{x})\right]}_{Z} + \underbrace{\frac{k}{m} \sum_{j=1}^k E^2_j}_{A} + \underbrace{\frac{2k}{m} \sum_{j=1}^k E_j\cdot \left(c_j(\vec{x})-\frac{m}{k}\right)}_{B} - \underbrace{\frac{k}{m} \sum_{j=1}^k E_j}_{C}
\end{align*}

We will show that, when $\bD=\bU$, this sum is below threshold $\tau$ with probability $\geq 9/10$. When $\tv{\bD}{\bU} >\alpha$, it is below $\tau$ with probability $\leq 29/50$. We will prove this result using a series of claims whose proofs we defer to the end.

We start with results that hold regardless of the identity of $\bD$. Recall that the noise scale $\ell_j$ computed in $\cA_\ut$ is distributed as $\ell_j \sim \Pois(\lambda)$. Invoking Lemma \ref{lem:poi_con} once more,
\begin{equation}
\label{eq:scale-concentration}
\P{}{\exists j \in [k] ~ \ell_j > 2\lambda } < \frac{1}{50}.
\end{equation}
Likewise, we obtain bounds on the moments of $A$, $B$, and $C$:
\begin{claim}
\label{claim:moments}
\begin{gather}
\ex{}{A} = \frac{k}{4m}\sum_{j=1}^k \ell_j ~~~ \ex{}{B} = 0, ~~~ \ex{}{C} = 0 \label{eq:ex}\\
\Var{}{A} \leq \frac{k^2}{8m^2}\sum_{j=1}^k \ell^2_j  ~~~  \Var{}{C} = \frac{k^2}{4m^2} \sum_{j=1}^k \ell_j \label{eq:var}
\end{gather}
\end{claim}
This is proven in Appendix \ref{app:testing-claims}. We now split into cases based on $\bD$.

\underline{Case 1}: $\bD = \bU$. We therefore want to upper bound $Z'$, starting with bounds on the moments of $Z$ and $B$.
\begin{claim}
\label{claim:moments-uniform}
Sample $n\sim\Pois(m)$ and $\vec{x}\sim \bU^n$. In an execution of $\cP^\ut_{\lambda,\alpha}(\vec{x})$,
\begin{gather}
\ex{}{Z} \leq \frac{\alpha^2 m}{500}\\
\Var{}{Z} \leq \frac{\alpha^4m^2}{500000} ~~~ \Var{}{B} = \frac{k}{m} \sum_{j=1}^k \ell_j \label{eq:var-uniform}
\end{gather}
\end{claim}
This is also proven in Appendix \ref{app:testing-claims}. To upper bound each of $Z,A,B,C$, we apply Chebyshev's inequality to inequalities \eqref{eq:ex} through \eqref{eq:var-uniform}. Then we apply \eqref{eq:scale-concentration} and a union bound to conclude that the following holds except with probability $\leq 5/50 = 1/10$:
\begin{align*}
Z' &= Z + A + B - C \\
&< \underbrace{\frac{\alpha^2 m}{500} + \sqrt{\frac{\alpha^4m^2}{500000} \cdot 50}}_{Z} + \underbrace{\frac{k}{4m}\sum_{j=1}^k \ell_j + \sqrt{\frac{k^2}{8m^2}\sum_{j=1}^k \ell^2_j \cdot 50} }_{A} + \underbrace{\sqrt{ \frac{k}{m} \sum_{j=1}^k \ell_j \cdot 50}}_{B} + \underbrace{\sqrt{\frac{k^2}{4m^2} \sum_{j=1}^k \ell_j \cdot 50}}_{C} \\
&\leq \frac{3\alpha^2 m}{250} + \frac{k}{4m}\sum_{j=1}^k \ell_j + \frac{5k}{2m}\sqrt{\sum_{j=1}^k \ell^2_j } + \sqrt{ \frac{50k}{m} \sum_{j=1}^k \ell_j} + \frac{5k}{\sqrt{2}m} \sqrt{ \sum_{j=1}^k \ell_j} \\
&\leq \frac{3\alpha^2 m}{250} + \frac{k^2}{2m} \cdot \lambda + \frac{5k^{3/2}}{m} \cdot \lambda + \frac{10 k}{m^{1/2}} \cdot \lambda^{1/2} + \frac{5k^{3/2}}{m} \lambda^{1/2} \stepcounter{equation} \tag{\theequation} \label{eq:z'-upper}
\end{align*}
We will set $\tau$ to the above quantity. Thus, given uniform samples, the protocol has probability at least 9/10 of correctly answering ``uniform.''

\underline{Case 2}: $\tv{\bD}{\bU} > \alpha$. We use the analysis of~\citet{ADK15} to get
\begin{align}
    \ex{}{Z} &\geq \frac{\alpha^2 m}{5} \label{eq:ex-z-non}\\
    \Var{}{Z} &\leq \frac{\ex{}{Z}^2}{100} \label{eq:var-z-non}
\end{align}
This allows us to use Chebyshev's inequality to obtain a lower bound on $Z$. Claim \ref{claim:moments} allows us to do the same for $A$ and $C$. It remains to find a lower bound on $B$. Note that $E_j$ is symmetrically distributed with mean zero. We now invoke the following technical claim, proven in Appendix \ref{app:testing-claims}:

\begin{claim}
\label{clm:symmetric-mean-zero}
Let $E_1, \dots, E_k$ be independent random variables where each $E_j$ is symmetrically distributed over the integers with mean zero. For any $d_1,\dots,d_k \in \R$, the random variable $\sum_{j=1}^k E_j \cdot d_j$ is symmetrically distributed with mean zero.
\end{claim}
Hence, $B$ is symmetric with mean 0 so that $\P{}{B\geq 0}\geq 1/2$.

As stated before, we apply Chebyshev's inequality to the pair \eqref{eq:ex}, \eqref{eq:var} and the pair \eqref{eq:ex-z-non}, \eqref{eq:var-z-non}. Then we use \eqref{eq:scale-concentration} and $\P{}{B\geq 0}\geq 1/2$ to conclude that the following is true except with probability $\leq 1/2 + 4/50=29/50$:
\begin{align}
Z' &> \ex{}{Z} - \sqrt{\frac{\ex{}{Z}^2}{100} \cdot 50} + \underbrace{\frac{k}{4m}\sum_{j=1}^k \ell_j - \sqrt{\frac{k^2}{8m^2}\sum_{j=1}^k \ell^2_j \cdot 50} }_{A} - \underbrace{\sqrt{\frac{k^2}{4m^2} \sum_{j=1}^k \ell_j \cdot 50} }_{C} \nonumber \\
&> \frac{29\alpha^2 m}{100} + \frac{k}{4m}\sum_{j=1}^k \ell_j - \frac{5k}{2m}\sqrt{\sum_{j=1}^k \ell^2_j } - \frac{5k}{\sqrt{2}m} \sqrt{ \sum_{j=1}^k \ell_j} \\
&\geq \frac{29\alpha^2 m}{100} + \frac{k^2}{2m} \cdot \lambda - \frac{5k^{3/2}}{m} \cdot \lambda - \frac{5k^{3/2}}{m} \lambda^{1/2} \label{eq:z'-lower}
\end{align}

We will prove that $\eqref{eq:z'-lower}> \eqref{eq:z'-upper}$. Since $\eqref{eq:z'-upper} = \tau$, this means the probability of erroneously reporting ``uniform'' is at most $29/50$.
\begin{align*}
&\eqref{eq:z'-lower} - \eqref{eq:z'-upper}\\
={}& \frac{29\alpha^2 m}{100} +\frac{k^2}{2m} \cdot \lambda - \frac{5k^{3/2}}{m} \cdot \lambda - \frac{5k^{3/2}}{m} \lambda^{1/2}\\
&- \frac{3\alpha^2 m}{250} - \frac{k^2}{2m} \cdot \lambda - \frac{5k^{3/2}}{m} \cdot \lambda - \frac{10 k}{m^{1/2}} \cdot \lambda^{1/2} - \frac{5k^{3/2}}{m} \lambda^{1/2}\\
={}& \frac{278\alpha^2 m}{1000} - \frac{10k^{3/2}}{m} \cdot \lambda - \frac{10 k}{m^{1/2}} \cdot \lambda^{1/2} - \frac{10k^{3/2}}{m} \lambda^{1/2}
\end{align*}
For some constant $h_{1}$, if $m> \frac{h_{1}}{\alpha}\cdot k^{3/4} \cdot \lambda^{1/2}$ then $\frac{278\alpha^2 m}{2000} - \frac{10k^{3/2}}{m} \cdot \lambda - \frac{10k^{3/2}}{m} \lambda^{1/2} > 0$. For some constant $h_{2}$, if $m > \frac{h_{2}}{\alpha^{4/3}} k^{2/3} \lambda^{1/3}$ then $\frac{278\alpha^2 m}{2000} - \frac{10 k}{m^{1/2}} \cdot \lambda^{1/2} > 0$. Adding these inequalities gives us $\eqref{eq:z'-lower} - \eqref{eq:z'-upper}>0$ which concludes the proof.
\end{proof}

\subsubsection{Final Uniformity Tester}
\label{sec:final-ut}
We now use a technique from~\citet{ACHST19} and ~\citet{AJM19} (itself a generalization of a similar technique from~\citet{ACFT19}) to reduce the sample complexity dependence on $k$ from $k^{3/4}$ to $k^{2/3}$. The idea is to reduce the size of the data universe $[k]$ by grouping random elements and then performing the test on the smaller universe $[\hat{k}]$. The randomized grouping also reduces testing distance --- partitions may group together elements with non-uniform mass to produce a group with near-uniform overall mass, thus hiding some of the original distance --- but the reduction in universe size outweighs this side effect.

We first introduce some notation. Given a partition $G$ of $[k]$ into $G_1,\dots,G_{\hat{k}}$ , let $\bD_G$ denote the distribution over $[\hat{k}]$ such that the probability of sampling $\hat{j}$ from $\bD_G$ is the probability that $j \in G_{\hat{j}}$ for $j\sim \bD$. Formally, $\P{}{\bD_G = \hat{j}} = \sum_{j \in G_{\hat{j} }} \P{}{\bD = j}$.

\begin{lemma}[Domain Compression \cite{ACHST19, AJM19}]
\label{lem:partition}
Let $\bD$ be a distribution over $[k]$ such that $\tv{\bD}{\bU}=\alpha$. If $G$ is a uniformly random of $[k]$ into $\hat{k}$ groups, then with probability $\geq 1/954$ over $G$,
$$\tv{\bD_G}{\bU} \geq \alpha \cdot \frac{\sqrt{\hat{k}} }{477\sqrt{10k}}.$$
\end{lemma}

Applying this trick to reduce $[k]$ to $[\hat k]$ and then running our initial protocol $\cP_\ut = (\cR_\ut,\cA_\ut)$ on $[\hat k]$ with distance parameter $\hat \alpha = \alpha \tfrac{\sqrt{\hat k}}{477\sqrt{10k}}$ gives our final uniformity tester. The given asymptotic bound requires different values of $\hat k$ depending on parameter settings; these appear in the proof.

\begin{theorem}
\label{thm:ut-app-upper-bound}
Let $\gamma \in (0,1]$, $\eps > 0$, and $\alpha, \delta \in (0,1)$. There exists a protocol that is $(2\eps, 8\delta^\gamma, \gamma)$-robustly shuffle private and, when $\eps = O(1)$ and $\delta=o(1)$, solves $\alpha$-uniformity testing with sample complexity
$$m = O\left(\left(\frac{k^{2/3}}{\alpha^{4/3} \eps^{2/3}} + \frac{k^{1/2}}{\alpha \eps} + \frac{k^{1/2}}{\alpha^2} \right) \cdot \ln^{1/2} \left(\frac{k}{\delta}\right) \right).$$
\end{theorem}
\begin{proof}
The protocol assigns $\hat{k}$ according to the following rule:
\[
\hat{k} = \begin{cases}
2 & \textrm{if} \frac{k^{2/3}\eps^{4/3}}{\alpha^{4/3}} < 2\\
k & \textrm{if} \frac{k^{2/3}\eps^{4/3}}{\alpha^{4/3}} > k\\
\frac{k^{2/3}\eps^{4/3}}{\alpha^{4/3}} & \textrm{otherwise}
\end{cases}
\]

The users and analyzer determine the partition $G$ using public randomness. Users execute $\cR_\ut$ as if the data were drawn from $\bD_G$ and the analyzer executes $\cA_\ut$. Privacy is immediate from Theorem \ref{thm:ut-upper-bound-prelim}, so it remains to argue that the protocol is accurate for large enough $m$.

\underline{Sample Complexity}: Recall that we set $\hat \alpha = \alpha \tfrac{\sqrt{\hat k}}{477\sqrt{10k}}$. We thus bound $m$ using the sample complexity guarantee of Theorem~\ref{thm:ut-upper-bound-prelim}:
\begin{align*}
m &= O\left(\frac{\hat k^{3/4}}{\hat \alpha \eps} \ln^{1/2}\left(\frac{\hat k}{\delta}\right) + \frac{\hat k^{2/3}}{\hat \alpha^{4/3} \eps^{2/3}} \ln^{1/3}\left(\frac{\hat k}{\delta}\right) + \frac{\hat k^{1/2}}{\hat\alpha^2}  \right)\\
	&= O\left(\left(\frac{\hat k^{3/4}}{\hat \alpha \eps} + \frac{\hat k^{2/3}}{\hat \alpha^{4/3} \eps^{2/3}} + \frac{\hat k^{1/2}}{\hat\alpha^2}  \right)\cdot \ln^{1/2}\left(\frac{k}{\delta}\right) \right) \tag{$\hat k \leq k$}\\
	&= O\left( \left( \underbrace{\frac{k^{1/2} \hat{k}^{1/4} }{\alpha\eps } }_{T_1} + \underbrace{\frac{k^{2/3}}{ \alpha^{4/3}\eps^{2/3}} }_{T_2} + \underbrace{\frac{k}{\alpha^2\hat{k}^{1/2}} }_{T_3} \right) \ln^{1/2} \left(\frac{k}{\delta}\right) \right) \tag{Value of $\hat \alpha$}
\end{align*}
We split into cases based on $\hat k$. They yield, respectively, the $\tfrac{k^{1/2}}{\alpha \eps}$, $\tfrac{k^{1/2}}{\alpha^2}$, and $\tfrac{k^{2/3}}{\alpha^{4/3}\eps^{2/3}}$ terms for $O(T_1+T_2+T_3)$.

\underline{Case 1}: $\hat k = 2$. Then $\tfrac{k^{2/3} \eps^{4/3}}{\alpha^{4/3}}  < 2$, so $k^{1/2} = O\left(\tfrac{\alpha}{\eps}\right)$ and $k^{1/6} = O\left(\tfrac{\alpha^{1/3}}{\eps^{1/3}}\right)$. Thus,
\begin{align*}
	T_1 + T_2 + T_3 =&\ O\left(\frac{k^{1/2}}{\alpha \eps} + \frac{k^{1/2}\cdot k^{1/6}}{\alpha^{4/3}\eps^{2/3}} + \frac{k^{1/2}k^{1/2}}{\alpha^2} \right) \\
	=&\ O\left(\frac{k^{1/2}}{\alpha \eps} + \frac{k^{1/2}}{\alpha \eps} + \frac{k^{1/2}}{\alpha \eps}\right) \\
	=&\ O\left(\frac{k^{1/2}}{\alpha\eps}\right).
\end{align*}

\underline{Case 2}: $\hat k = k$. This means $k < \tfrac{k^{2/3}\eps^{4/3}}{\alpha^{4/3}}$, so $k^{3/4} < \tfrac{k^{1/2}\eps}{\alpha}$ and $k^{1/6} < \tfrac{\eps^{2/3}}{\alpha^{2/3}}$. Thus,
\begin{align*}
T_1 + T_2 + T_3 =&\ O\left(\frac{k^{3/4}}{\alpha \eps} + \frac{k^{2/3}}{\alpha^{4/3}\eps^{2/3}} + \frac{k^{1/2}}{\alpha^2} \right) \\
	=&\ O\left(\frac{k^{1/2}}{\alpha^2} + \frac{k^{2/3}}{\alpha^{4/3}\eps^{2/3}} \right) \\
	=&\ O\left(\frac{k^{1/2}}{\alpha^2} + \frac{k^{1/2} \cdot k ^{1/6}}{\alpha^{4/3}\eps^{2/3}} \right) \\
	=&\ O\left(\frac{k^{1/2}}{\alpha^2}\right).
\end{align*}

\underline{Case 3}: $\hat{k} = \tfrac{k^{2/3} \eps^{4/3}}{\alpha^{4/3}}$. By substitution,
\begin{align*}
T_1 + T_2 + T_3 =&\ O\left(\frac{k^{1/2} (k^{2/3} \eps^{4/3} \alpha^{-4/3})^{1/4} }{\alpha\eps } + \frac{k^{2/3}}{\alpha^{4/3}\eps^{2/3}} + \frac{k}{\alpha^2 (k^{2/3} \eps^{4/3} \alpha^{-4/3})^{1/2}}\right) \\
	=&\ O\left(\frac{k^{2/3} }{\alpha^{4/3}\eps^{2/3} } \right)
\end{align*}

\underline{Correctness}: We apply Theorem \ref{thm:ut-upper-bound-prelim} (with a constant number of repetitions) to get
\begin{itemize}
    \item If $\vec{x}\sim \bD_G^n$ where $\bD_G=\bU$, then $\P{}{\cP_\ut(\vec{x}) = \textrm{``uniform''}} \geq 9539/9540$, and
    \item If $\vec{x}\sim \bD_G^n$ where $\tv{\bD_G}{\bU} > \hat{\alpha}$, then $\P{}{\cP_\ut(\vec{x}) = \textrm{``not uniform''}} \geq 9539/9540$
\end{itemize}
If $\bD=\bU$, then $\bD_G=\bU$ and so the probability of ``uniform'' is $\geq 9539/9540$. If $\tv{\bD}{\bU}>\alpha$ then by Lemma~\ref{lem:partition}, with probability $\geq 10/9540$, $\tv{\bD_G}{\bU} > \hat \alpha$. By Theorem~\ref{thm:ut-upper-bound-prelim}, with probability $\geq 5/9540$, the tester returns ``non-uniform''. This constant separation gives the overall testing guarantee.
\end{proof}

\subsection{Lower Bound for Robust Pure Shuffle Privacy}
To obtain a lower bound for robust pure shuffle privacy, we first show how to transform a robustly shuffle private uniformity tester into a pan-private uniformity tester. The main idea of this transformation is the same as that for distinct elements. We initialize the pan-private algorithm's state using dummy data, handle new stream elements as shuffle protocol users contributing to a growing pool of (repeatedly) shuffled messages, and add more dummy data to the internal state at the end of the stream. Here, the dummy data consists of samples from a uniform distribution. This has the effect of diluting the true samples and worsens the testing accuracy, but to a controlled extent. Pseudocode for this procedure appears in Algorithm~\ref{alg:ut-transformation}.

\begin{algorithm}

\caption{$\cQ_\cP$, an $(\eps,\delta)$-pan-private tester built from a $1/3$-robust $(\eps,\delta)$-shuffle private tester}

\label{alg:ut-transformation}

%\vb{need to shuffle in messages}

\KwIn{Data stream $\vec{x}\in[k]^{n/3}$; shuffle-private uniformity tester $\cP=(\cR,\cA)$}

\KwOut{Decision in $\{ \textrm{``uniform''}, \textrm{``not uniform''} \}$}

Draw uniform samples $\vec{x}_{\bU} \sim \bU^{n/3}$

Initialize internal state $I_0 \gets (\cS \circ \cR^{n/3})(\vec{x}_{\bU})$

Draw $n' \sim \Bin(n,2/9)$

Set $n' \gets \min(n',n/3)$

\For{$i\in[n/3]$}{
    \If{$i \leq n'$}{Set $I_i \gets \cS(I_{i-1} ,\cR(x_i))$}
    \Else{Set $I_i \gets \cS(I_{i-1} , \cR(\bU))$}
}

Draw (new) uniform samples $\vec{x}_{\bU} \sim \bU^{n/3}$

Set final state $\vec{y} \gets \cS(I_{n/3}, \cR^{n/3}(\vec{x}_{\bU}))$

\Return{$\cA(\vec{y})$}

\end{algorithm}

\begin{lemma}
\label{lemma:ut-transformation}
Let $\cP=(\cR, \cA)$ be an $(\eps,\delta,1/3)$-robustly shuffle private $\alpha$-uniformity tester with sample complexity $n$. If $n$ is larger than some absolute constant, then $\cQ_\cP$ described in Algorithm~\ref{alg:ut-transformation} is an $(\eps,\delta)$-pan-private algorithm that solves $\tfrac{9\alpha}{2}$ uniformity testing with sample complexity $n/3$.
\end{lemma}
\begin{proof}
\underline{Privacy}: The privacy argument is nearly identical to that used to prove Lemma~\ref{lem:de-transformation}. The only difference is that the dummy data now consists of uniform samples. However, since this dummy data is still independent of the true data, we can apply the same argument to translate robust shuffle privacy into pan-privacy.

\underline{Accuracy}: Assuming that $\cP$ errs with probability at most $1/3$, we will prove that the algorithm $\cQ_\cP$ errs with probability at most 1/3 on data drawn from $\bU$. Then we bound the error probability by 1/2 when data is drawn from a distribution that is $\tfrac{9\alpha}{2}$-far from $\bU$. This constant separation implies a valid tester.

Let $\cP(\vec{x})$ denote the output distribution of the original shuffle private uniformity tester, and let $\cQ_\cP(\vec{x})$ denote the output distribution for the pan-private uniformity tester given in Algorithm~\ref{alg:ut-transformation} on stream $\vec{x}$. If $\vec{x}$ consists of uniform samples, then $\cQ_\cP(\bU^{n/3}) = \cA((\cS\circ \cR)(\bU^n)) = \cP(\bU^n)$ and so
$$\P{\vec x, \cQ_\cP}{\cQ_\cP \text{ outputs ``not uniform''}} = \P{\vec x, \cP}{\cP \text{ outputs ``not uniform''}} \leq 1/3.$$

Having upper bounded the probability that $\cQ_\cP$ errs on uniform samples, we now control the probability that it errs on non-uniform samples. Suppose $\vec{x}$ consists of samples from $\bD$ where $\tv{\bD}{\bU} > \tfrac{9\alpha}{2}$. Define the mixture distribution $\bD_{2/9} := \frac{2}{9}\cdot \bD + \frac{7}{9}\cdot \bU$. Then $\tv{\bD_{2/9}}{\bU} > \alpha$ and so $\P{}{\cP(\bD^n_{2/9}) = \textrm{``uniform''}}\leq 1/3$. We will now show that $\cP(\bD^n_{2/9})$ and $\cQ_\cP(\bD^{n/3})$ are statistically close enough that $\cQ_\cP$ has a bounded error probability on non-uniform samples as well.

In $n$ samples from $\bD_{2/9}$, the number of samples drawn from $\bD$ is distributed as $\Bin(n,2/9)$. By a binomial Chernoff bound, for $n > \sqrt{\ln(12)}$, $\P{}{\Bin(n,2/9) > n/3} < \tfrac{1}{6}$. Thus the probability that $n'$ is not distributed as $\Bin(n,2/9)$ is less than $1/6$. In turn, the distance between
\[
\cP(\bD^n_{2/9})= \cA(\cS(\overbrace{ \underbrace{\cR(\bD), \dots, \cR(\bD)}_{\Bin(n,2/9) ~\textrm{copies}}, \cR(\bU),\dots,\cR(\bU) }^{n~\textrm{terms}}))
\]
and
\[
\cQ_{\cP}(\bD^{n/3})= \cA(\cS(\overbrace{ \underbrace{\cR(\bD), \dots, \cR(\bD)}_{n' ~\textrm{copies}}, \cR(\bU),\dots,\cR(\bU) }^{n~\textrm{terms}} ))
\]
is less than $1/6$ as well. Tracing back, we have shown that given samples from sufficiently non-uniform $\bD$, $\P{}{\cQ_\cP \text{ outputs ``uniform''}} \leq 1/3 + 1/6 = 1/2$.
\end{proof}

Next, we recall the pan-private lower bound for uniformity testing.
\begin{lemma}[Theorem 3 from~\citet{AJM19}]
\label{lem:ut-lower-bound-pan}
For $\eps=O(1)$ and $\alpha < 1/2$, any $\eps$-pan-private $\alpha$-uniformity tester has sample complexity
$$\Omega\left(\frac{k^{2/3}}{\alpha^{4/3} \eps^{2/3}}+ \frac{\sqrt{k}}{\alpha^2}+  \frac{1}{\alpha \eps}\right).$$
\end{lemma}
Together, Lemmas~\ref{lemma:ut-transformation} and~\ref{lem:ut-lower-bound-pan} imply our lower bound for pure robustly shuffle private uniformity testing. 
\begin{theorem}
\label{thm:ut-pure-lower-bound}
For $\eps=O(1)$ and $\alpha < 1/9$, any $(\eps,0,1/3)$-robustly shuffle private protocol $\alpha$-uniformity tester has sample complexity
$$\Omega \left( \frac{k^{2/3}}{\alpha^{4/3} \eps^{2/3}}+ \frac{\sqrt{k}}{\alpha^2} + \frac{1}{\alpha \eps} \right).$$
\end{theorem}
Note that the lower bound of Theorem~\ref{thm:ut-pure-lower-bound} is not directly comparable to the upper bound of Theorem~\ref{thm:ut-app-upper-bound}, since the former only applies to robust pure shuffle privacy while the latter only satisfies robust approximate shuffle privacy. This in turn is because a lower bound is only known for pure pan-privacy. We note that Lemma~\ref{lemma:ut-transformation} would also apply to an approximate pan-private lower bound.
\section{Pan-private Histograms}
\label{sec:histogram}
We now depart from our previous results by transforming a shuffle private protocol into a pan-private algorithm. Specifically,~\citet{BC20} gave a shuffle private protocol for estimating histograms with error independent of the domain size. Their protocol relies on adding noise from a binomial distribution for privacy; in Appendix \ref{sec:robust-hist}, we show that this strategy is also robust. We then show that a pan-private analogue using binomial noise achieves the same error.

The main building block in our histogram algorithm is a counting algorithm $\cQ_\zsum$ (Algorithm \ref{alg:zerocount-pan}). To count a sum of bits, $\cQ_\zsum$ adds binomial noise to its counter before beginning the stream, updates the counter deterministically for each bit in the stream, adds more binomial noise at the end of the stream, and releases the resulting noisy count. Because binomial noise is bounded, we can always bound the final error.

\begin{algorithm}
\caption{An online algorithm $\cQ_\zsum$ for binary sums}
\label{alg:zerocount-pan}

\KwIn{Data stream $\vec{x}\in\zo^n$; parameter $\lambda \in \N$}
\KwOut{$z \in \N$}

Draw $I_0 \sim \Bin(\lambda,1/2)$

\For{$i\in[n]$}{
	$I_i \gets I_{i-1} + x_i$
}

Draw $\eta \sim \Bin(\lambda,1/2)$

Set $\tc \gets I_n + \eta$

\Return{$\tc$}
\end{algorithm}

\begin{theorem}
\label{theorem:zerocount-pan}
	Given $\eps > 0$, $\delta \in (0,1)$, and $\lambda \geq 20\cdot \left(\frac{e^\eps+1}{e^\eps - 1}\right)^2 \ln\left(\tfrac{2}{\delta}\right)$, $\cQ_\zsum$ is $(\eps,\delta)$-pan-private and computes binary sums with error $\leq 2\lambda = O\left(\tfrac{1}{\eps^2}\log\left(\tfrac{1}{\delta}\right)\right)$.
\end{theorem}
\begin{proof}
	\underline{Privacy}: The basic idea of the proof is that, by Lemma~\ref{lem:b_noise}, the first draw of binomial noise ensures privacy the internal state view, and the second draw of binomial noise ensures privacy for the output view. Substituting in Lemma~\ref{lem:b_noise}, the privacy analysis is almost identical to that for Lemma~\ref{lem:de-transformation}.

	\underline{Accuracy}: For any stream $\vec{x}$, we have $\tc = \sum_{i=1}^n x_i + \Bin(2\lambda,1/2)$, which has support ranging from $\sum_{i=1}^n x_i$ to $\sum_{i=1}^n x_i + 2\lambda$. 
\end{proof}

With $\cQ_\zsum$ in hand, our histogram algorithm is simple: we run $\cQ_\zsum$ on each value in the data domain. Crucially, changing one stream element only changes the true value of at most two bins in the histogram, so by composition the resulting algorithm's privacy guarantee is within a factor of two of $\cQ_\zsum$. Since each count is $2\lambda$-accurate, we get the same $\ell_\infty$ accuracy. This improves over a naive solution that adds Laplace noise to each bin, which incurs a logarithmic dependence on the domain size.

\begin{theorem}
\label{theorem:hist-pan}
Given $\eps > 0$, $\delta \in (0,1)$, and $\lambda \geq 20\cdot \left(\frac{e^\eps+1}{e^\eps - 1}\right)^2 \ln\left(\tfrac{2}{\delta}\right)$, the histogram algorithm described above is $(2\eps, 2\delta)$-pan-private and computes a histogram with $\ell_\infty$ error at most $2\lambda = O\left(\tfrac{1}{\eps^2}\log\left(\tfrac{1}{\delta}\right)\right)$.
\end{theorem}

Pan-privacy thus inherits the same separations from (sequentially interactive) local privacy as those outlined by Balcer and Cheu~\cite{BC20} for shuffle privacy. For example, using histogram gives an $(\eps,\delta)$-pan-private solution to pointer-chasing on length-$\ell$ vectors using $O\left(\tfrac{1}{\eps^2}\log\left(\tfrac{1}{\delta}\right)\right)$ samples. Sequentially interactive $(\eps,\delta)$-local privacy, however, requires $\Omega(\ell)$ samples~\cite{JMR20}.
\section{Conclusion and Further Questions}
\label{sec:conc}
Our results suggest a relationship between robust shuffle privacy and pan-privacy. In addition to the straightforward problems of closing gaps in our upper and lower bounds, we conclude with more general questions:
\begin{enumerate}
	\item \emph{When can we convert robust shuffle private protocols to pan-private algorithms?} Both of our robust shuffle private lower bounds rely on this kind of conversion. In particularly, they use the fact that the underlying problem is resilient to ``fake'' data. For example, in the course of converting a robust shuffle private distinct elements protocol to a pan-private one, it was important that adding many draws from $\cR(1)$ --- i.e., adding many copies of 1 to the data --- only changed the true answer by at most 1. Similarly, when converting the robust shuffle private uniformity tester to a pan-private one, adding fake uniform samples only diluted the original testing distance, so the resulting pan-private tester was still useful. It is not clear which problems do or do not have this property, or whether this property is necessary in general.
	\item \emph{When can we convert pan-private algorithms to robust shuffle private protocols?} Our robust shuffle private distinct elements counter is structurally similar to its pan-private counterpart~\cite{DNPRY10}: both essentially break distinct elements into a sum of noisy $\mor$s and then de-bias the result. The only difference is in the distributed noise generation of our shuffle protocol. Similarly, our robust shuffle private uniformity tester differs from the pan-private version~\cite{AJM19} only in the kind of noise added. Is there a generic structural condition that allows for this kind of transformation?
	\item \emph{Can we separate robust shuffle privacy and pan-privacy?} Proving that robust shuffle privacy must obtain worse performance for some problem requires a lower bound that holds for robust shuffle privacy but not pan-privacy. Unfortunately, almost all multi-message robust shuffle privacy lower bounds are either (1) imported directly from central privacy or (2) imported from pan-privacy. In either case, such lower bounds also apply to pan-privacy and thus do not yield a separation. One exception is the summation lower bound of~\citet{GGKMPV20}, but it holds only if each user is limited to $O(\poly\log(n))$ communication. In the other direction, the only pan-private lower bounds that do not also hold for central privacy are for distinct elements and uniformity testing, but we have shown that they do not give a (polynomial in domain size) separation here.
\end{enumerate}

\newpage

\section*{Acknowledgments}
We thank Cl\'{e}ment Canonne for simplifying the form of the uniformity testing lower bound and Adam Smith for useful discussions regarding the pan-privacy definition. We also thank Jonathan Ullman, Vikrant Singhal and Salil Vadhan for general commentary.

\bibliographystyle{plainnat}
\bibliography{references}

\begin{thebibliography}{30}
\providecommand{\natexlab}[1]{#1}
\providecommand{\url}[1]{\texttt{#1}}
\expandafter\ifx\csname urlstyle\endcsname\relax
  \providecommand{\doi}[1]{doi: #1}\else
  \providecommand{\doi}{doi: \begingroup \urlstyle{rm}\Url}\fi

\bibitem[Acharya et~al.(2015)Acharya, Daskalakis, and Kamath]{ADK15}
Jayadev Acharya, Constantinos Daskalakis, and Gautam Kamath.
\newblock Optimal testing for properties of distributions.
\newblock In \emph{Neural Information Processing Systems (NIPS)}, 2015.

\bibitem[Acharya et~al.(2018)Acharya, Sun, and Zhang]{ASZ18}
Jayadev Acharya, Ziteng Sun, and Huanyu Zhang.
\newblock Differentially private testing of identity and closeness of discrete
  distributions.
\newblock In \emph{Neural Information Processing Systems (NeurIPS)}, 2018.

\bibitem[Acharya et~al.(2019{\natexlab{a}})Acharya, Canonne, Freitag, and
  Tyagi]{ACFT19}
Jayadev Acharya, Cl{\'e}ment Canonne, Cody Freitag, and Himanshu Tyagi.
\newblock Test without trust: Optimal locally private distribution testing.
\newblock In \emph{International Conference on Artificial Intelligence and
  Statistics (AISTATS)}, 2019{\natexlab{a}}.

\bibitem[Acharya et~al.(2019{\natexlab{b}})Acharya, Canonne, Han, Sun, and
  Tyagi]{ACHST19}
Jayadev Acharya, Cl{\'{e}}ment~L. Canonne, Yanjun Han, Ziteng Sun, and Himanshu
  Tyagi.
\newblock Domain compression and its application to randomness-optimal
  distributed goodness-of-fit.
\newblock \emph{CoRR}, abs/1907.08743, 2019{\natexlab{b}}.

\bibitem[Amin et~al.(2020)Amin, Joseph, and Mao]{AJM19}
Kareem Amin, Matthew Joseph, and Jieming Mao.
\newblock Pan-private uniformity testing.
\newblock In \emph{Conference on Learning Theory (COLT)}, 2020.

\bibitem[Avent et~al.(2017)Avent, Korolova, Zeber, Hovden, and
  Livshits]{AKZHL17}
Brendan Avent, Aleksandra Korolova, David Zeber, Torgeir Hovden, and Benjamin
  Livshits.
\newblock Blender: Enabling local search with a hybrid differential privacy
  model.
\newblock In \emph{USENIX Security Symposium (USENIX)}, 2017.

\bibitem[Balcer and Cheu(2020)]{BC20}
Victor Balcer and Albert Cheu.
\newblock Separating local and shuffled differential privacy via histograms.
\newblock In \emph{Information Theoretic Cryptography (ITC)}, 2020.

\bibitem[Balle et~al.(2019{\natexlab{a}})Balle, Bell, Gasc{\'{o}}n, and
  Nissim]{BBGN19}
Borja Balle, James Bell, Adri{\`{a}} Gasc{\'{o}}n, and Kobbi Nissim.
\newblock The privacy blanket of the shuffle model.
\newblock In \emph{International Cryptology Conference (CRYPTO)},
  2019{\natexlab{a}}.

\bibitem[Balle et~al.(2019{\natexlab{b}})Balle, Bell, Gasc{\'{o}}n, and
  Nissim]{BBGN19b}
Borja Balle, James Bell, Adri{\`{a}} Gasc{\'{o}}n, and Kobbi Nissim.
\newblock Differentially private summation with multi-message shuffling.
\newblock \emph{arXiv preprint arXiv:1906.09116}, 2019{\natexlab{b}}.

\bibitem[Balle et~al.(2020)Balle, Bell, Gasc{\'{o}}n, and Nissim]{BBGN20}
Borja Balle, James Bell, Adri{\`{a}} Gasc{\'{o}}n, and Kobbi Nissim.
\newblock Private summation in the multi-message shuffle model.
\newblock \emph{arXiv preprint arXiv:2002.00817}, 2020.

\bibitem[Beimel et~al.(2008)Beimel, Nissim, and Omri]{BNO08}
Amos Beimel, Kobbi Nissim, and Eran Omri.
\newblock Distributed private data analysis: Simultaneously solving how and
  what.
\newblock In \emph{International Cryptology Conference (CRYPTO)}, 2008.

\bibitem[Beimel et~al.(2010)Beimel, Kasiviswanathan, and Nissim]{BBKN10}
Amos Beimel, Shiva~Prasad Kasiviswanathan, and Kobbi Nissim.
\newblock Bounds on the sample complexity for private learning and private data
  release.
\newblock In \emph{Theory of Cryptography Conference}, pages 437--454.
  Springer, 2010.

\bibitem[Bittau et~al.(2017)Bittau, Erlingsson, Maniatis, Mironov, Raghunathan,
  Lie, Rudominer, Kode, Tinnes, and Seefeld]{BEMMR+17}
Andrea Bittau, \'{U}lfar Erlingsson, Petros Maniatis, Ilya Mironov, Ananth
  Raghunathan, David Lie, Mitch Rudominer, Ushasree Kode, Julien Tinnes, and
  Bernhard Seefeld.
\newblock Prochlo: Strong privacy for analytics in the crowd.
\newblock In \emph{Symposium on Operating Systems Principles (SOSP)}, 2017.

\bibitem[Bun et~al.(2016)Bun, Nissim, and Stemmer]{BNS16}
Mark Bun, Kobbi Nissim, and Uri Stemmer.
\newblock Simultaneous private learning of multiple concepts.
\newblock In \emph{Innovations in Theoretical Computer Science (ITCS)}, 2016.

\bibitem[Cai et~al.(2017)Cai, Daskalakis, and Kamath]{CDK17}
Bryan Cai, Constantinos Daskalakis, and Gautam Kamath.
\newblock Priv'it: private and sample efficient identity testing.
\newblock In \emph{International Conference on Machine Learning (ICML)}, 2017.

\bibitem[Canonne(2017)]{C17}
Cl{\'{e}}ment~L. Canonne.
\newblock A short note on poisson tail bounds, 2017.
\newblock URL
  \url{http://www.cs.columbia.edu/~ccanonne/files/misc/2017-poissonconcentration.pdf}.

\bibitem[Cheu et~al.(2019)Cheu, Smith, Ullman, Zeber, and Zhilyaev]{CSUZZ19}
Albert Cheu, Adam Smith, Jonathan Ullman, David Zeber, and Maxim Zhilyaev.
\newblock Distributed differential privacy via shuffling.
\newblock In \emph{Annual International Conference on the Theory and
  Applications of Cryptographic Techniques (CRYPTO)}, 2019.

\bibitem[Dwork et~al.(2006{\natexlab{a}})Dwork, Kenthapadi, McSherry, Mironov,
  and Naor]{DKMMN06}
Cynthia Dwork, Krishnaram Kenthapadi, Frank McSherry, Ilya Mironov, and Moni
  Naor.
\newblock Our data, ourselves: Privacy via distributed noise generation.
\newblock In \emph{Conference on the Theory and Applications of Cryptographic
  Techniques (EUROCRYPT)}, 2006{\natexlab{a}}.

\bibitem[Dwork et~al.(2006{\natexlab{b}})Dwork, McSherry, Nissim, and
  Smith]{DMNS06}
Cynthia Dwork, Frank McSherry, Kobbi Nissim, and Adam Smith.
\newblock Calibrating noise to sensitivity in private data analysis.
\newblock In \emph{Theory of Cryptography Conference (TCC)},
  2006{\natexlab{b}}.

\bibitem[Dwork et~al.(2010)Dwork, Naor, Pitassi, Rothblum, and
  Yekhanin]{DNPRY10}
Cynthia Dwork, Moni Naor, Toniann Pitassi, Guy~N Rothblum, and Sergey Yekhanin.
\newblock Pan-private streaming algorithms.
\newblock In \emph{Innovations in Computer Science (ICS)}, 2010.

\bibitem[Dwork et~al.(2014)Dwork, Roth, et~al.]{DR14}
Cynthia Dwork, Aaron Roth, et~al.
\newblock The algorithmic foundations of differential privacy.
\newblock \emph{Foundations and Trends{\textregistered} in Theoretical Computer
  Science}, 2014.

\bibitem[Erlingsson et~al.(2019)Erlingsson, Feldman, Mironov, Raghunathan,
  Talwar, and Thakurta]{EFMRTT19}
{\'U}lfar Erlingsson, Vitaly Feldman, Ilya Mironov, Ananth Raghunathan, Kunal
  Talwar, and Abhradeep Thakurta.
\newblock Amplification by shuffling: From local to central differential
  privacy via anonymity.
\newblock In \emph{Symposium on Discrete Algorithms (SODA)}, 2019.

\bibitem[Ghazi et~al.(2019{\natexlab{a}})Ghazi, Golowich, Kumar, Pagh, and
  Velingker]{GGKPV19}
Badih Ghazi, Noah Golowich, Ravi Kumar, Rasmus Pagh, and Ameya Velingker.
\newblock On the power of multiple anonymous messages.
\newblock \emph{Arxiv}, abs/1908.11358, 2019{\natexlab{a}}.

\bibitem[Ghazi et~al.(2019{\natexlab{b}})Ghazi, Pagh, and Velingker]{GPV19}
Badih Ghazi, Rasmus Pagh, and Ameya Velingker.
\newblock Scalable and differentially private distributed aggregation in the
  shuffled model.
\newblock \emph{CoRR}, abs/1906.08320, 2019{\natexlab{b}}.

\bibitem[Ghazi et~al.(2020)Ghazi, Golowich, Kumar, Manurangsi, Pagh, and
  Velingker]{GGKMPV20}
Badih Ghazi, Noah Golowich, Ravi Kumar, Pasin Manurangsi, Rasmus Pagh, and
  Ameya Velingker.
\newblock Pure differentially private summation from anonymous messages.
\newblock In \emph{Information Theoretic Cryptography (ITC)}, 2020.

\bibitem[Hardt and Talwar(2010)]{HT10}
Moritz Hardt and Kunal Talwar.
\newblock On the geometry of differential privacy.
\newblock In \emph{Proceedings of the forty-second ACM symposium on Theory of
  computing}, pages 705--714, 2010.

\bibitem[Ishai et~al.(2006)Ishai, Kushilevitz, Ostrovsky, and Sahai]{IKOS06}
Yuval Ishai, Eyal Kushilevitz, Rafail Ostrovsky, and Amit Sahai.
\newblock Cryptography from anonymity.
\newblock In \emph{Foundations of Computer Science (FOCS)}, 2006.

\bibitem[Joseph et~al.(2020)Joseph, Mao, and Roth]{JMR20}
Matthew Joseph, Jieming Mao, and Aaron Roth.
\newblock Exponential separations in local differential privacy.
\newblock In \emph{Symposium on Discrete Algorithms (SODA)}, 2020.

\bibitem[Kasiviswanathan et~al.(2011)Kasiviswanathan, Lee, Nissim,
  Raskhodnikova, and Smith]{KLNRS11}
Shiva~Prasad Kasiviswanathan, Homin~K Lee, Kobbi Nissim, Sofya Raskhodnikova,
  and Adam Smith.
\newblock What can we learn privately?
\newblock \emph{SIAM Journal on Computing}, 2011.

\bibitem[Mir et~al.(2011)Mir, Muthukrishnan, Nikolov, and Wright]{MMNW11}
Darakhshan Mir, Shan Muthukrishnan, Aleksandar Nikolov, and Rebecca~N Wright.
\newblock Pan-private algorithms via statistics on sketches.
\newblock In \emph{Principles of Database Systems (PODS)}. ACM, 2011.

\end{thebibliography}

\newpage

\appendix
\section{Appendix}
\label{sec:app}

\subsection{Proofs for Robust Shuffle Private Distinct Elements Upper Bound}
\label{subsec:distinct-lemmas}
In Section \ref{sec:distinct-upper}, we used two auxiliary lemmas to argue that $\cP_\de$ is robustly private. Here, we give their proofs.

\begin{lemma}[Restatement of Lemma \ref{lem:distinct-ber}]
	Let $n \in \N$, $\gamma \in (0,1]$, and $p \in [0, 1/2]$. Define $p' = \tfrac{1-(1-2p)^{1/n}}{2}$. Then given i.i.d.\ $X_1, \ldots, X_{\gamma n} \sim \mathbf{Ber}(p')$, $X = \sum_{i=1}^{\gamma n} X_i \textrm{ mod } 2$ is identically distributed with $$\mathbf{Ber}\left(\frac{1-(1-2p)^\gamma}{2}\right).$$
\end{lemma}
\begin{proof}
	For all $i \in [\gamma n]$, define $Y_i = 1 - 2X_i \in \{\pm 1\}$, and define $Y = \prod_{i=1}^{\gamma n}$. Then
	$$\P{}{X=1} = \P{}{Y = -1} = \frac{1-\E{}{Y}}{2}.$$
	Now, by independence, we rewrite
	$$\E{}{Y} = \prod_{i=1}^{\gamma n} \E{}{Y_i} = \prod_{i=1}^{\gamma n} (1 - 2\E{}{X_i}) = (1-2p')^{\gamma n}$$
	to get $\P{}{X=1} = \tfrac{1 - (1-2p')^{\gamma n}}{2}$. Plugging in the value $p' = \tfrac{1-(1-2p)^{1/n}}{2}$ gives our result.
\end{proof}

\begin{lemma}[Restatement of Lemma \ref{lem:distinct-gamma}]
	For $\gamma \in (0,1]$, $\tfrac{2^\gamma}{2^\gamma-1} \leq \tfrac{2}{\gamma}$.
\end{lemma}

\begin{proof}
	Since $\gamma \cdot \tfrac{2^\gamma}{2^\gamma-1} = 2$ at $\gamma = 1$, it suffices to show that $\gamma \cdot \tfrac{2^\gamma}{2^\gamma-1}$ is nondecreasing on $(0,1]$. Its first derivative is
	$$\frac{2^\gamma(2^\gamma - \gamma \ln(2) - 1)}{(2^\gamma-1)^2}$$
	so  it suffices to show that $2^\gamma - \gamma \ln(2) - 1 \geq 0$. Equality holds for $\gamma = 0$ and its first derivative is $(2^\gamma-1)\ln(2) \geq 0$ for $\gamma > 0$. Thus $2^\gamma - \gamma\ln(2) - 1 \geq 0$ as desired.
\end{proof}
\subsection{The Distinct Elements Problem Over a Large Universe}
\label{subsec:app_distinct_2}
When $k \ge n^2$, the distinct elements protocol $\cP_\de$ from Section \ref{sec:distinct-upper} does not provide a meaningful notion of accuracy. Using public randomness to select a hash function, we can obtain a protocol for distinct elements with error $O(n^{2/3})$ by having each user hash their input to a smaller domain before running the distinct element protocol $\cP_\de$. Note that hashing is strictly for utility and is not used to add privacy.

First, we give a high probability bound on the error from estimating the number of distinct elements by counting the number of distinct elements after hashing.
\begin{lemma}\label{lem:distinct-hash}
Let $k,k' \in \N$ such that $k \ge k'$.
Let $h$ be sampled uniformly from a 2-universal hash family $\cH$ mapping $[k]$ to $[k']$.
Let $S \subseteq [k]$ and $S' = \{s' \in [k'] \mid \exists s \in S \text{ s.t. } h(s) = s'\} \subseteq [k']$.
Then for all $\beta \in (0,1)$,
$$\P{h}{|S| - |S'| \geq \frac{|S|^2}{\beta k'}} \leq \beta.$$
\end{lemma}
\begin{proof}
Trivially, $|S| \ge |S'|$.
Let $X = |\{(s,s') \in S^2 \mid s<s' \text{ and } h(s) = h(s')\}|$, i.e. the number of collisions when hashing the set $S$.
Notice that $|S'| = |\{s' \in S \mid \forall (s \in S \text{ s.t. } s < s') ~ h(s) \neq h(s')\}|$.
This implies $|S| - |S'| = |\{s' \in S \mid \exists s \in S \text{ s.t. } s < s' \text{ and } h(s) = h(s')\}| \le X$.
Since $h$ is sampled uniformly from a 2-universal hash family, $\mathbb{E}[X] \le |S|^2/k'$. 
The result follows by Markov's inequality.
\end{proof}

\begin{corollary}
Let constant $c \ge 1$.
Let $n,k \in \N$ such that $k \ge \lceil cn^{4/3} \rceil$.
Let $\eps > 0$ and $\delta \in (0,1)$.
Let $h$ be sampled uniformly from a 2-universal hash family $\cH$ mapping $[k]$ to $[\lceil cn^{4/3}\rceil]$.
Then the protocol $\cP_\mathsf{HDE} = (\cR_\mathsf{DE} \circ h, \cA_\mathsf{DE})$ for $n$ users %has the following properties:
\begin{enumerate}[label=\Roman*.]
\item is $(2\eps + 2\ln(1/\gamma), 2\delta/\gamma, \gamma)$-robustly shuffle private;
\item solves the $(\alpha, \beta)$-distinct elements problem for
\begin{align*}
\alpha = \frac{2n^{2/3}}{c\beta} + \frac{e^\eps}{e^\eps-1}\cdot\sqrt{2(n^{4/3}+1)\ln\left(4/\beta\right)};
\end{align*}
\item requires each user to communicate at most $O(n^{4/3}\log(n(e^\eps+1)/\delta))$ messages of length $O(\log n)$.
\end{enumerate}
\end{corollary}
\begin{proof}[Proof of Privacy (Part I)]
Fix $h \in \cH$.
For $\vec{x} \in [k]^n$, let $h(\vec{x}) = (h(x_1), \ldots, h(x_n))$.
Let $\vec{x}, \vec{x}\,' \in [k]^n$ be neighboring datasets.
Then $h(\vec{x})$ and $h(\vec{x}\,')$ are also neighboring datasets, and privacy follows from Theorem \ref{thm:de_upper} Part I.
\end{proof}
\begin{proof}[Proof of Accuracy (Part II)]
By Lemma \ref{lem:distinct-hash} and Theorem \ref{thm:de_upper} Part II,
\begin{align*}
&\mathbb{P}\!\left[\big|\cP_\mathsf{HDE}(\vec{x}) - D(\vec{x})\big| \ge \frac{2n^{2/3}}{c\beta} + \frac{e^\eps}{e^\eps-1}\cdot\sqrt{2(n^{4/3}+1)\ln\left(4/\beta\right)}\right]\\
&\le \mathbb{P}\!\left[\big|\cP_\mathsf{DE}(h(\vec{x})) - D(h(\vec{x}))\big| + \big|D(h(\vec{x}))- D(\vec{x})\big| \ge \frac{2n^{2/3}}{c\beta} + \frac{e^\eps}{e^\eps-1}\cdot\sqrt{2(n^{4/3}+1)\ln\left(4/\beta\right)} \right]\\
&\le \mathbb{P}\!\left[\big|\cP_\mathsf{DE}(h(\vec{x})) - D(h(\vec{x}))\big| \ge \frac{e^\eps}{e^\eps-1}\cdot\sqrt{2(n^{4/3}+1)\ln\left(4/\beta\right)}\right] + \mathbb{P}\!\left[\big|D(h(\vec{x}))- D(\vec{x})\big| \ge \frac{2n^{2/3}}{c\beta}\right]\\
&\le \beta. \qedhere
\end{align*}
\end{proof}
\subsection{Technical Claims for Unifomity Testing}
\label{app:testing-claims}
Here, we provide proofs for the technical claims made in the proof of Theorem \ref{thm:ut-upper-bound-prelim}.

\begin{proof}[Proof of Claim \ref{claim:moments}]
Recall that we defined $E_j = c_j(\vec{y}) - c_j(\vec{x})$. is drawn from $\Bin(\ell_j, 1/2) -\ell_j/2$ whose first four moments are $0, \frac{\ell_j}{4}, 0, \frac{3\ell^2_j}{16} - \frac{\ell_j}{8}$. The expectation of $A$ immediately follows from linearity and the second moment of $E_j$:
$$ \ex{}{A} = \frac{k}{m} \sum_{j=1}^k \ex{}{E_j^2} = \frac{k}{4m} \sum_{j=1}^k \ell_j $$

The expectations of $B,C$ are 0 due to $\E{}{E_j}=0$:
\begin{align*}
\ex{}{B} &= \frac{2k}{m} \sum_{j=1}^k \ex{}{E_j\cdot (c_j(\vec{x})-m/k)}\\
    &= \frac{2k}{m} \sum_{j=1}^k \ex{}{E_j} \cdot \ex{}{c_j(\vec{x})-m/k} \tag{Independence} \\
    &= 0 \\
\ex{}{C} &= \frac{k}{m} \sum_{j=1}^k \ex{}{E_j} = 0
\end{align*}

The variance calculations follow essentially the same recipe:
\begin{align*}
\Var{}{A} &= \frac{k^2}{m^2} \sum_{j=1}^k \Var{}{ E_j^2 } \tag{Independence}\\
    &= \frac{k^2}{m^2}\sum_{j=1}^k (\ex{}{E_j^4} - \ex{}{E_j^2}^2)\\
    &= \frac{k^2}{m^2}\sum_{j=1}^k (\frac{3\ell^2_j}{16} - \frac{\ell_j}{8} - \frac{\ell_j^2}{16}) \tag{$4^\mathrm{th}$ \& $2^\mathrm{nd}$ moments}\\
    &\leq \frac{k^2}{8m^2}\sum_{j=1}^k \ell^2_j \\
\Var{}{C} &= \frac{k^2}{m^2} \sum_{j=1}^k \Var{}{E_j} \tag{Independence}\\
    &= \frac{k^2}{4m^2} \sum_{j=1}^k \ell_j \qedhere
\end{align*}
\end{proof}

\begin{proof}[Proof of Claim \ref{claim:moments-uniform}]
As observed in \cite{AJM19}, the analysis by \cite{ADK15} implies that
\begin{align*}
    \ex{}{Z} &\leq \frac{\alpha^2 m}{500}\\
    \Var{}{Z} &\leq \frac{\alpha^4m^2}{500000} 
\end{align*}
Also, we have
\begin{align*}
\Var{}{B} &= \frac{4k^2}{m^2} \sum_{j=1}^k \Var{}{E_j \cdot (c_j(\vec{x}) - m/k)} \tag{Independence} \\
    &= \frac{4k^2}{m^2} \sum_{j=1}^k \ex{}{E_j^2} \cdot \ex{}{(c_j(\vec{x}) - m/k)^2} - \ex{}{E_j}^2 \cdot \ex{}{c_j(\vec{x}) - m/k}^2 \tag{Independence} \\
    &= \frac{k^2}{m^2} \sum_{j=1}^k \ell_j \cdot \ex{}{(c_j(\vec{x}) - m/k)^2} \stepcounter{equation} \tag{\theequation} \label{eq:var-b}\\
    &= \frac{k^2}{m^2} \sum_{j=1}^k \ell_j \cdot \frac{m}{k} \tag{$\bD=\bU$}\\
    &= \frac{k}{m} \sum_{j=1}^k \ell_j
\end{align*}
where \eqref{eq:var-b} follows from the fact that $E_j \sim \Bin(\ell_j,1/2)$.
\end{proof}

\begin{proof}[Proof of Claim \ref{clm:symmetric-mean-zero}]
By linearity of expectation, the mean is zero. Now we show that the variable is symmetric. We first argue that each term is symmetric: for any $v\neq 0$,
\begin{align*}
\P{}{E_j\cdot d_j = v} &=  \P{}{E_j = v / d_j}\\
    &= \P{}{E_j = -v / d_j} \\
    &= \P{}{E_j\cdot d_j = -v}
\end{align*}
It remains to argue the convolution of two symmetric distributions $\bT,\bT'$ is symmetric:
\begin{align*}
\P{t\sim\bT,t'\sim\bT'}{t+t'=v} &= \sum_{u\in \Z} \P{}{t=u} \cdot \P{}{t'=v-u} \\
    &= \sum_{u \in \Z} \P{}{t=-u} \cdot \P{}{t'=-(v-u)} \\
    &= \P{}{t+ t'=-v} \qedhere
\end{align*}
\end{proof}

\subsection{Robustness of Counting Protocol by Balcer \& Cheu}
\label{sec:robust-hist}
In order to compute histograms with error $O(\log(1/\delta)/\eps^2)$,~\citet{BC20} presented a shuffle private protocol $\cP^\mathrm{zsum}_{\eps,\delta}$ for binary sums. Here, we prove that their binary sum protocol is robust shuffle private. As written, it only ensures privacy for $\eps \le 1$, but this is a limitation that is lifted by replacing instances of $\eps$ with $\sqrt{5} \cdot (e^\eps-1) / (e^\eps+1)$.
\begin{claim}
\label{claim:robust_bc}
For any $\eps > 0$, $\gamma \in (01,]$, and $\delta \in (0,1)$ with $n \geq 20 \left(\tfrac{e^\eps+1}{e^\eps-1}\right)^2\ln\left(\tfrac{2}{\delta}\right)$ users, $\cP^\mathrm{zsum}_{\eps,\delta}$  is $\left(\eps, 2\left(\tfrac{\delta}{2}\right)^\gamma, \gamma\right)$-robustly shuffle private.
\end{claim}
\begin{proof}
In $\cP^\mathrm{zsum}_{\eps,\delta}$, each user reports their true bit and a bit drawn from $\mathbf{Ber}(p)$ where
\begin{align*}
p = 1 - \frac{10}{n}\cdot\left(\frac{e^\eps+1}{e^\eps-1}\right)^2\cdot\ln\left(\frac{2}{\delta}\right).
\end{align*}
For every $\gamma' \ge \gamma$ such that $\gamma'n \in \Z$, if only $\gamma' n$ users run this protocol, then on input $\vec{x}$, the output of the shuffler is a post-processing of the central algorithm that samples $\eta \sim \mathbf{Bin}(\gamma'n, p)$ and outputs $\left(\sum_{i=1}^{\gamma' n} x_i\right) + \eta$. It therefore suffices to prove a privacy guarantee for this quantity. We remark that Lemma \ref{lem:b_noise} is a special case of a more general lemma:
\begin{lemma}[Appendix C \cite{GGKPV19}]
Let $n,\ell \in \N$, $p \in (0,1)$ and let $f: \cX^n \rightarrow \Z$ be a 1-sensitive function, i.e. $|f(\vec{x}) - f(\vec{x}\,')| \le 1$ for all neighboring datasets $\vec{x},\vec{x}\,' \in \cX^n$.
For any $\eps > 0$ and $\delta \in (0,1]$, if
$$\ell \cdot \min(p, 1-p) \geq 10\left(\frac{e^\eps+1}{e^\eps-1}\right)^2\ln\left(\frac{2}{\delta}\right),$$
then the algorithm that on input $\vec{x}$ samples $\eta \sim \mathbf{Bin}(\ell, p)$ and outputs $f(\vec{x}) + \eta$ is $(\eps, \delta)$-differentially private.
\end{lemma}
We apply this lemma after making the following observation:
\begin{align*}
\gamma' n \cdot \min(p, 1-p)
&\geq \gamma' \cdot 10 \cdot\left(\frac{e^\eps+1}{e^\eps-1}\right)^2\cdot\ln\left(\frac{2}{\delta}\right)\\
&\ge \gamma \cdot 10 \cdot\left(\frac{e^\eps+1}{e^\eps-1}\right)^2\cdot\ln\left(\frac{2}{\delta}\right)\\
&= 10 \cdot\left(\frac{e^\eps+1}{e^\eps-1}\right)^2\cdot\ln\left(\frac{2}{2(\delta/2)^\gamma}\right).
\end{align*}
Therefore, $\cP^\mathrm{zsum}_{\eps,\delta}$ is $(\eps, 2(\delta/2)^\gamma, \gamma)$-robustly shuffle private.
\end{proof}

\end{document}